\documentclass[11pt,letterpaper]{article}

\usepackage[margin=1.00in]{geometry}

\usepackage{amsmath, amsthm, amsfonts}

\usepackage{thmtools} 
\usepackage{thm-restate}
\usepackage{xspace} 

\usepackage{hyperref}
\usepackage{cleveref}

\usepackage{algorithm}
\usepackage{algpseudocode}

\usepackage{xcolor}
\usepackage{tikz,pgf}

\definecolor{darkgreen}{rgb}{0,0.5,0}
\definecolor{darkred}{rgb}{0.4,0,0}
\hypersetup{
    colorlinks=true,
    linkcolor=darkred,
    citecolor=darkgreen,
    filecolor=black,
    urlcolor=[rgb]{0,0.1,0.5},
    pdftitle={Simpler and More General Distributed Coloring Based on
    Simple List Defective Coloring Algorithms},
    pdfauthor={Marc Fuchs, Fabian Kuhn}
}

\newcommand{\para}[1]{\paragraph{#1}}

\newtheorem{theorem}{Theorem}[section]
\newtheorem{lemma}[theorem]{Lemma}
\newtheorem{claim}[theorem]{Claim}

\newtheorem{definition}{Definition}[section]

\newcommand{\calC}{\ensuremath{\mathcal{C}}}

\newcommand{\Nout}{N^{\mathit{out}}}

\newcommand{\CONGEST}{\ensuremath{\mathsf{CONGEST}}\xspace}

\newcommand{\LOCAL}{\ensuremath{\mathsf{LOCAL}}\xspace}

\newcommand{\eps}{\varepsilon}
\renewcommand{\epsilon}{\varepsilon}
\newcommand{\poly}{\operatorname{\text{{\rmfamily poly}}}}

\newcommand{\set}[1]{\left\{#1\right\}}

\DeclareMathOperator{\polylog}{\poly\log}

\newcommand{\hide}[1]{}

\renewcommand{\phi}{\varphi}

\begin{document}

\title{Simpler and More General Distributed Coloring Based on
  Simple List Defective Coloring Algorithms\footnote{This work was partially supported by the German Research Foundation (DFG) under the project number 491819048.}}

\date{}

\author{
   Marc Fuchs \\
   \small{University of Freiburg} \\
   \small{marc.fuchs@cs.uni-freiburg.de}
   \and
   Fabian Kuhn \\
   \small{University of Freiburg} \\
   \small{kuhn@cs.uni-freiburg.de}
   }

\maketitle

\begin{abstract}
  In this paper, we give list coloring variants of simple iterative defective coloring algorithms. Formally, in a list defective
  coloring instance, each node $v$ of a graph is given a list $L_v$ of  colors and a list of allowed defects $d_v(x)$ for the colors. Each
  node $v$ needs to be colored with a color $x\in L_v$ such that at most $d_v(x)$ neighbors (or outneighbors) of $v$ also pick the same color $x$. We adapt two existing simple defective coloring algorithms to the list setting.

  For a defect parameter $d$, it is known that by making two sweeps in opposite order over the nodes of an edge-oriented graph with maximum outdegree $\beta$, one can compute a coloring with $O(\beta^2/d^2)$ colors such that every node has at most $d$ outneighbors of the same color. We generalize this and show that if all nodes have lists of size $p^2$ and $\forall v : \sum_{x\in L_v} (d_v(x)+1)>p\cdot \beta$, we can make two sweeps of the nodes such that at the end, each node $v$ has chosen a color $x\in L_v$ for which at most $d_v(x)$ outneighbors of $v$ are colored with color $x$. Our algorithm is simpler and computationally significantly more efficient than existing algorithms for similar list defective coloring problems. We show that the above result can in particular be used to obtain an alternative $\tilde{O}(\sqrt{\Delta}) + O(\log^* n)$-round algorithm for the $(\Delta+1)$-coloring problem in the \CONGEST model. The algorithm is simpler and also computationally more efficient than the only other such algorithm from [Fuchs, Kuhn; DISC '23].

  The neighborhood independence $\theta$ of a graph is the maximum number of pairwise non-adjacent neighbors of some node of the graph. Notable examples of graphs of bounded neighborhood independence are the line graphs of bounded-rank hypergraphs. It is well known that by doing a single sweep over the nodes of a graph of neighborhood independence $\theta$, one can compute a $d$-defective coloring with $O(\theta\cdot \Delta/d)$ colors. We extend this approach to the list defective coloring setting and use it to obtain an efficient recursive coloring algorithm for graphs of bounded neighborhood independence $\theta$. In particular, if $\theta=O(1)$, we get an $(\log\Delta)^{O(\log\log\Delta)}+O(\log^* n)$-round algorithm to compute a $(\Delta+1)$-coloring. The algorithm generalizes an algorithm from [Balliu, Kuhn, Olivetti; PODC '20], which computes a $(2\Delta-1)$-edge coloring (and thus a $(\Delta+1)$-coloring in line graphs of graphs) with the same round complexity.
\end{abstract}

\section{Introduction and Related Work}
\label{sec:intro}

Distributed coloring is arguably the most studied problem in the area of distributed graph algorithms. In the most standard form, the task is to color the nodes of some $n$-node graph $G=(V,E)$, which also defines the network topology. The nodes of $G$ have unique $O(\log n)$-bit identifiers and they can exchange messages over the edges of $G$ in synchronous rounds. If the messages can be of arbitrary size, this setting is known as the \LOCAL model and if the message size is restricted to $O(\log n)$ bits, it is known as the \CONGEST model. The number of allowed colors is usually some function of the maximum degree $\Delta$. Of particular interest is the problem of coloring with $\Delta+1$ colors, i.e., with the number of colors used by a sequential greedy algorithm.

\paragraph{\textbf{State of the Art.}} We start by giving a brief summary of the previous work on distributed coloring algorithms. The work on distributed and parallel coloring started in the 1980s~\cite{Awerbuch89,cole86,goldberg88,Linial1987,Luby1986}. In a seminal paper, Linial showed that coloring a ring network with $O(1)$ colors (and thus coloring graphs with $f(\Delta)$ colors for any function $f$) requires $\Omega(\log^* n)$ rounds. In the same paper, Linial also gives a $O(\log^*n)$-round algorithm to color a graph with $O(\Delta^2)$ colors. Together with a simple color reduction scheme, one obtains a $(\Delta+1)$-coloring in $O(\Delta^2 + \log^* n)$ rounds. For bounded-degree graphs, we thus know that the complexity of distributed $(\Delta+1)$-coloring is $\Theta(\log^* n)$. Over the years, the dependency of the round complexity on $\Delta$ has been improved in multiple steps, leading to the current best algorithm, which has a round complexity of $O(\sqrt{\Delta\log\Delta} + \log^* n)$ in the \LOCAL model~\cite{goldberg88,Linial1987,KuhnW06,BarenboimE09,Kuhn2009,barenboim16sublinear,fraigniaud16local,BarenboimEG18,MausT20}. The fastest known algorithm does require the use of large messages. There is however a recent \CONGEST algorithm that almost matches the round complexity of the best \LOCAL algorithm and computes a $(\Delta+1)$-coloring in $O(\sqrt{\Delta}\cdot\log^2\Delta\cdot\log^6\log\Delta + \log^* n)$ rounds~\cite{FK23}.

Apart from analyzing the dependency on $\Delta$, there has naturally also been a lot of work that tried to optimize the round complexity as a function of the number of nodes $n$. Even with the simple first randomized algorithms of the 1980s, it is possible to $(\Delta+1)$-color a graph in only $O(\log n)$ rounds~\cite{Alon1986,Linial1987,Luby1986}. Classically, the fastest deterministic algorithms (as a function of $n$) were based on a generic tool known as network decomposition. Using the construction of \cite{Awerbuch89,panconesi96decomposition}, one gets a $2^{O(\sqrt{\log n})}$-round algorithm to compute a $(\Delta+1)$-coloring. A major breakthrough was achieved relatively recently by Rozho\v{n} and Ghaffari~\cite{Rozhon2020}, who showed that a network decompostion and thus also a $(\Delta+1)$-coloring can be computed deterministically in time polylogarithmic in $n$. By using more specialized methods, the best deterministic complexity of $(\Delta+1)$-coloring has subsequently been improved in several steps to $\tilde{O}(\log^2 n)$ in the \LOCAL model and $O(\log^2\Delta\cdot\log n)$ in the \CONGEST model~\cite{GGR2020,GhaffariKuhn21,GhaffariGrunau23}.\footnote{We use the notation $\tilde{O}(x)$ to denote functions of the form $x\cdot\poly\log x$.} In the last decade, there has also been a big push towards obtaining faster randomized $(\Delta+1)$-coloring algorithms. The current state of the art complexity for randomized $(\Delta+1)$-coloring is $\tilde{O}(\log^2\log n)$ in \LOCAL, $O(\log^3\log n)$ in \CONGEST, and even $O(\log^*n)$ if $\Delta=\Omega(\log^3 n)$ (both in \LOCAL and \CONGEST)~\cite{Barenboim2016,ChangLP18,Flin0HKN23,HalldorssonKMT21,HalldorssonKNT22,HalldorssonNT22,HarrisSS18}. 

\paragraph{\textbf{Defective Coloring.}} A key tool that is used in essentially all modern deterministic distributed coloring algorithms is some variant of \emph{defective coloring}. In its most standard from, defective coloring is a natural relaxation of graph coloring. For an integer parameter $d\geq 0$, a $c$-coloring of the nodes of a graph $G=(V,E)$ is called $d$-defective if each node has at most $d$ neighbors of the same color (i.e., if each color class induces a subgraph of degree at most $d$). Defective colorings have in particular been used in two ways. Decomposing a given graph into a small number of subgraphs of smaller degree allows to use divide-and-conquer approaches to distributed coloring (e.g., \cite{BalliuKO20,BalliuBKO22,BarenboimE09,BarenboimE10,barenboim11,Kuhn20,GhaffariKuhn21}) and also to speed up the distributed execution of some sequential greedy-like algorithms that can tolerate some local conflicts (e.g., \cite{BarenboimE09,BarenboimE10,barenboim11,Barenboim2016,fraigniaud16local,Kuhn20,GhaffariKuhn21}). Although defective coloring might seem to be a simple relaxation of proper graph coloring, it is algorithmically quite different. While it is known that every graph has a $d$-defective coloring with $\lceil\frac{\Delta+1}{d+1}\rceil$ colors~\cite{lovasz66}, we do not know a greedy-type algorithm that iterates over the nodes once or a small number of times and comes close to this. The best greedy-type algorithm is a simple algorithm that processes the nodes in two sweeps that iterate through the nodes in opposite order. In both sweeps, the node picks a color that minimizes the defect to the previously colored nodes in the current sweep and the final color is built as the cartesian product of the colors of the two sweeps. This allows to compute a $d$-defective coloring with $\lceil\frac{\Delta+1}{d+1}\rceil^2$ colors~\cite{BalliuHLOS19,BarenboimE09}.\footnote{By analyzing a bit more carefully, one can improve this bound by essentially a factor $2$ because not all possible pairs of colors of the $2$ sweeps are needed.} In the distributed setting, a $d$-defective coloring with $O\big(\big(\frac{\Delta}{d+1}\big)^2\big)$ colors can also be computed more efficiently in only $O(\log^* n)$ rounds~\cite{Kuhn2009} by adapting the $O(\log^* n)$-round $O(\Delta^2)$-coloring algorithm of Linial~\cite{Linial1987}. The only known way to efficiently obtain $d$-defective colorings with fewer colors is through the use of LLL methods~\cite{ChungPS14, FischerG17, GhaffariHK18}. All those algorithms inherently require deterministic and randomized round complexities of $\Omega(\log n)$ and $\Omega(\log\log n)$, respectively~\cite{BrandtFHKLRSU16, chang16exponential}. The distributed coloring literature therefore also considered weaker variants of defective coloring such as in particular arbdefective colorings. A $d$-arbdefective coloring of a graph is a coloring of the nodes together with an orientation of the monochromatic edges such that every node has at most $d$ outneighbors of the same color~\cite{BarenboimE10}. Unlike for the standard defective coloring problem, a $d$-arbdefective coloring with $\lceil\frac{\Delta+1}{d+1}\rceil$ colors can be computed by a simple sequential greedy algorithm and there are various efficient distributed algorithms to compute $d$-arbdefective colorings with $O\big(\frac{\Delta}{d+1}\big)$ colors~\cite{BarenboimE10,Barenboim2016,BarenboimEG18,GhaffariKuhn21,FK23}.

\paragraph{\textbf{List Defective Coloring.}} Another key idea for all recent distributed coloring algorithms is explicit use of list colorings: In a list coloring instance, every node $v$ obtains a list $L_v$ of colors as input and as output, each node $v$ has to choose some color from $L_v$ such that the graph is properly colored. In order for this problem to be greedily solvable, one requires that $|L_v|\geq \deg(v)+1$ for all $v$. Such (degree+1)-list coloring instance naturally occur when solving $(\Delta+1)$-coloring in phases and one has to extend an existing partial proper coloring of the nodes in each of the phases. While the use of defective colorings and list colorings are both crucial tools combining them is sometimes challenging. When using defective coloring in divide-and-conquer algorithms, one typically uses a defective coloring to divide the graph and one then also partitions the colors among those subgraphs so that they can be colored in parallel. Without global communication, it is clearly not possible to partition the colors such that the induced partition of each of the lists is close to a uniform partition.\footnote{Even with global communication, this is only possible if the lists are sufficiently large (as a function of $n$).} One can generalize the idea in the following way. Assume that all the lists $L_v$ consist of colors from some space of size $C$ (say from the set $\set{1,\dots,C}$). One can then partition the global color space into $p>1$ parts of size $\approx C/p$. This induced a partition of the lists that is not uniform. Each node therefore has to select one of the $p$ remaining color spaces in such a way that the number of conflicting neighbors decreases at a rate that is comparable to the decrease of the list size. This \emph{color space reduction} idea has been introduced in \cite{Kuhn20} and it has since then been used in \cite{BalliuKO20,BalliuBKO22,FK23,GhaffariKuhn21}. In \cite{FK23}, it is shown that the core task that needs to be solved can be phrased as a natural list extension of defective coloring.

Formally, in \emph{list defective coloring}, each node $v$ obtains a color list $L_v$ together with a defect function $d_v:L_v \to \mathbb{N}_0$ that assigns non-negative integer values to the colors in $L_v$. The coloring has to assign a color $x_v\in L_v$ to each node $v$ such that the number of neighboring nodes $u$ that get assigned color $x_u=x_v$ is at most $d_v(x_v)$. One can naturally extend this to \emph{list arbdefective colorings} by also requiring an orientation of the monochromatic edges and by requiring that if $v$ is colored with color $x_v$, it has at most $d_v(x_v)$ outneighbors of color $x_v$. The authors of \cite{FK23} further consider a variant of list defective coloring in graphs with a given edge orientation. As in list arbdefective colorings, an \emph{oriented list defective coloring} (short: \emph{OLDC}) also requires to color the nodes such that 
if $v$ is colored with color $x_v$, it has at most $d_v(x_v)$ outneighbors of color $x_v$. But here, the edge orientations are a part of the input and not a part of the output. The main technical result of \cite{FK23} is an algorithm to compute an oriented list defective coloring.  Consider an edge-oriented graph with maximum out-degree $\beta$ and an initial proper $q$-coloring of the nodes. Further, assume that for each node $v$, $\beta_v$ is the outdegree of $v$, $v$ has a color list $L_v\subseteq\set{1,\dots,C}$ and defects $d_v(x)$ for $x\in L_c$ such that
\[
  \sum_{x \in L_v} (d_v(x) + 1)^2 > \alpha\cdot\beta_v^2\cdot (\log\beta + \log\log C + \log\log q)\cdot\log^2\log\beta\cdot(\log\log\beta+\log\log q)
\]
for some constant $\alpha>0$. Then, an oriented list defective coloring with colors from the given lists can be computed deterministically in $O(\log\beta)$ rounds in the \LOCAL model. The algorithm of \cite{FK23} is based on an algorithm by Maus and Tonoyan~\cite{MausT20} to properly list color an oriented graph with maximum outdegree $\beta$ and lists of size $O(\beta^2(\log\beta +\log\log C + \log\log q))$. The required message size in the algorithms of \cite{MausT20} and \cite{FK23} is linear in the maximum list size and the main application of the algorithm in \cite{FK23} was to use color space reduction in order to reduce the required message size when properly coloring oriented graphs. This resulted in the first $\tilde{O}(\sqrt{\Delta})+O(\log^* n)$-round $(\Delta+1)$-coloring algorithm for the \CONGEST model. One of the contributions of the present paper is a simpler and computationally much more lightweight algorithm to solve essentially the same problems as the algorithms in \cite{FK23,MausT20}. 

\paragraph{\textbf{Bounded Neighborhood Independence.}} \label{para:IntroBoundedNeighborhoodIndependence}
While for general graphs, it is not known if $d$-defective coloring algorithms that require $O(\Delta/d)$-colors can be computed in time $f(\Delta)\cdot\log^* n$ (except if $d=\Theta(\Delta)$), there is one important family of graph for which such defective coloring algorithms are known to exist. The \emph{neighborhood independence} $\theta$ of a graph $G=(V,E)$ is the maximum independent set size of any one-hop neighborhood of $G$, that is, it is the maximum independence number of any of the induced subgraphs $G[N(v)]$ for $v\in V$. The family of bounded neighborhood independence graphs is the family of graphs for which $\theta=O(1)$. Important examples for graphs of bounded neighborhood independence are the line graphs of hypergraphs of bounded rank. Because in a line graph of rank $r$, the neighboring edges of each hyperedge can be partitioned into $r$ cliques (in the line graph), the neighborhood independence of the line graph can be at most $r$.

For graphs with neighborhood independence $\theta$, a $d$-defective coloring with $O(\theta\cdot\Delta/(d+1))$ colors can be computed greedily as follows. We iterate over the nodes in an arbitrary order $v_1,v_2,\dots,v_n$ and we orient the edges such that for all edges $\set{v_i,v_j}$ with $i<j$, the edge is oriented from $v_j$ to $v_i$ (i.e., towards the earlier colored node). When a node $v$ picks its color, the colors of all outneighbors are already fixed and if we have $C$ colors available, $v$ can thus pick a color $x$ for such there are at most $\lfloor \Delta/C \rfloor$ outneighbors that are colored with color $x$. Let $N_x(v)$ be the set of neighbors of $v$ that are colored with color $x$ and let $G[N_x(v)]$ be the subgraph of $G$ induced by $N_x(v)$. We thus know that $G[N_x(v)]$ has an outdegree $\lfloor \Delta/C \rfloor$-orientation and therefore $G[N_x(v)]$ can be colored with at most $2\lfloor \Delta/C \rfloor+1$ colors. Because $G$ has neighborhood independence at most $\theta$, this implies that $|N_x(v)|$ and thus the defect of $v$ is at most $(2\lfloor \Delta/C \rfloor+1)\theta$.

In the context of distributed graph coloring, this relation has been used explicitly in \cite{barenboim11,FischerGK17,Kuhn20} and it for example also has been used for the special case of line graphs in \cite{BalliuKO20,BalliuBKO22}. In \cite{barenboim11}, the authors described an efficient distributed implementation of the above algorithm and apply it recursively to obtain $O(\Delta^{1+\eps})$-colorings for constant parameters $\eps>0$ in $\poly\log\Delta + O(\log^* n)$ rounds or $\Delta\cdot 2^{O(\sqrt{\log \Delta})}$-colorings in $2^{O(\sqrt{\log\Delta})}+O(\log^* n)$ rounds (everything under the assumption that the neighborhood independence $\theta=O(1)$). This has been generalized to list colorings in \cite{Kuhn20} and as a result of this, one can get a $(\Delta+1)$-coloring in time $2^{O(\sqrt{\log\Delta})}+O(\log^* n)$ in graphs of bounded neighborhood independence. For the special case of $(2\Delta-1)$-edge coloring, this was improved significantly in \cite{BalliuKO20,BalliuBKO22}. In \cite{BalliuBKO22}, it is shown that a $(2\Delta-1)$-edge coloring can be computed in time $O(\log\Delta)^{O(\log\log\Delta)}+O(\log^* n)$ in the \LOCAL model and in \cite{BalliuBKO22}, this was even improved to $O(\log^{12}\Delta + \log^* n)$ rounds. The algorithms for \cite{BalliuKO20,BalliuBKO22} however critically rely on having line graphs. The algorithms use the structure of line graphs explicitly in their description and they therefore cannot be extended to the general family of graphs of bounded neighborhood independence. The algorithm of \cite{BalliuKO20} can most likely be extended to the case of line graphs of bounded rank hypergraphs, the algorithm of \cite{BalliuBKO22} even critically depends on having line graphs of graphs (i.e., hypergraphs of rank $2$). As one of our main results, we obtain the same result as \cite{BalliuBKO22}, but for general bounded neighborhood independence graphs (cf.~\Cref{thm:mainNeighborhood}). Our algorithm has a similar highlevel recursive structure, but is otherwise different.

\subsection{Our Contributions}
\label{sec:contrib}

In the following, we describe our technical contributions in detail. We start with the list defective coloring version of the distributed Two-Sweep algorithm and its applications.

\paragraph{\textbf{Oriented List Defective Coloring Algorithm and Implications.}}

As our first main result, we give a list version of the simple Two-Sweep algorithm for computing defective colorings. Formally, we prove the following theorem.

\begin{restatable}{theorem}{restatefirst}\label{thm:2sweepalgo}
  There is a deterministic distributed algorithm $\mathcal{A}$ with the following properties.
  Let $G=(V,E)$ be a graph that is equipped with a proper $q$-vertex coloring and with an edge orientation. For every node $v\in V$, assume that $\beta_v$ is the outdegree of $v$. Further, assume that we are given an oriented list defective coloring instance for $G$ with color lists $L_v$ and defect functions $d_v:L_v\to \mathbb{N}_0$ for all $v\in V$. Let $p\geq 1$ be an integer parameter, let $\eps\geq 0$, and assume that
  \[
    \forall v \in V\,:\, \sum_{x\in L_v}(d_v(x)+1) > (1+\eps)\cdot\max\set{p,\frac{|L_v|}{p}}\cdot \beta_v.
  \]
  Then $\mathcal{A}$ solves the given oriented list defective coloring instance in $O\big(\min\set{q, (p/\eps)^2 + \log^*q}\big)$ rounds.\footnote{We slightly abuse notation and allow $\eps=0$, in which case we assume that $\min\set{q, (p/0)^2+\log^* q}=q$.} The algorithm requires the nodes to forward their initial color first and later exchange a list containing $p$ of the colors of $L_v$. 
\end{restatable}

One of the most important open problems in the context of defective coloring is to determine the combinations of defect $d$, number of colors $C$, and maximum degree $\Delta$ (or maximum outdegree $\beta$) such that a $d$-defective $C$-coloring can be computed in time at most $f(\Delta)\cdot \log^* n$ (or $f(\beta)\cdot\log^* n$). The above algorithm gives list versions for such defective colorings that are very close to the best trade-offs known for standard defective coloring. For example when coloring with $3$ colors, \cite{BalliuHLOS19}, it is shown that a $d$-defective $3$-coloring can be computed in time $O(\Delta+\log^* n)$ if $d\geq (2\Delta-4)/3$. We generalize this to arbitrary list defective and even oriented list defective coloring. We get a list $d$-defective $3$-coloring in time $O(\Delta+\log^*n)$ whenever $d>(2\Delta-3)/3$.

The algorithm for $\eps=0$ is based on a simple algorithm that iterates over the $q$ initial node colors twice. Once in increasing order and afterwards in decreasing order. In the first sweep over the nodes, each node $v$ selects a subset $S_v\subseteq L_v$ of size $|S_v|=\min\set{p, |L_v|}$ of its colors. The set $S_v$ is chosen so that the sum of defects of the colors in $S_v$ minus the sum of occurrences of those colors among the sets $S_u$ of outneighbors $u$ that have already chosen $S_u$ is maximized. In the second phase, the node chooses a color from $S_v$ for which it can guarantee the defect requirement. In the algorithm for $\eps>0$, the algorithm first computes a standard defective coloring from \cite{KawarabayashiS18,Kuhn2009} and it then applies the Two-Sweep algorithm on the subgraph induced by the monochromatic edges of this defective coloring.

\paragraph{Comparison to \cite{FK23,MausT20}.} To simplify the comparison with the algorithm of \cite{FK23}, we consider a scenario where the defects of all nodes and colors are equal to $d\geq 1$. The algorithm of \cite{FK23} then needs lists of size $\Omega\big(\big(\frac{\beta_v}{d}\big)^2\cdot (\log\beta_v+\log\log C)\big)$, while when choosing $p=\beta/d$, the algorithm of \Cref{thm:2sweepalgo} uses slightly smaller lists of size $O(p^2)=O\big(\big(\frac{\beta_v}{d}\big)^2\big)$, but it also has a round complexity of $O\big(p^2 + \log^* q\big)$ (for $\eps=\Theta(1)$). The time complexity is thus only comparable when the number of colors per list is small.

\paragraph{Computational complexity.} Apart from being simpler (in the algorithms of \cite{FK23,MausT20}, each node $v$ has to compute a subset of $2^{2^{L_v}}$ that satisfies certain properties), the main advantage of our algorithm compared to the one in \cite{FK23,MausT20} is the computational complexity of the internal computations at the nodes. As discussed in Appendix C of \cite{FK23full}, i.e., in the full version of \cite{FK23}, the internal complexity at each node is more than exponential in the maximum list size. The computational internal complexity of the nodes in our algorithm is nearly linear in $\Delta$ times the maximum list size. The main step for a node $v$ is aggregate the lists $S_u$ of the outneighbors $u$ and to then sort the colors in $L_v$ according to how good they are compared to colors in the sets $S_u$ of the outneighbors of $u$. The defective coloring algorithm of \cite{KawarabayashiS18,Kuhn2009} requires very little computation. Its most expensive step is to evaluate all the $O(1/\eps)$ point-value pairs of $\beta_v+1$ different polynomials over $\mathbb{F}_{O(1/\eps)}$ of degree $O(\log_{1/\eps} q)$.

\paragraph{List coloring with bounded outdegree.} It has been known since the late 1980s that an edge-oriented graph with maximum outdegree $\beta$ can be colored with $O(\beta^2)$ colors in $O(\log^* n)$ rounds~\cite{Linial1987}. For $\beta=O(1)$, we extend this result to solving list coloring instances with lists $L_v$ of size $|L_v|=\Theta(\beta^2)$. In $O(\log^* n)$ rounds, one can first use the algorithm of \cite{Linial1987} to compute an $O(\beta^2)$-coloring that we will then use as the initial $q$-coloring of our algorithm. By choosing $p=\beta+1$ and setting all defect values to $0$, we then get an $O(\beta^2 + \log^* n)$-round list coloring algorithm to properly color with lists $L_v$ of size $|L_v| \geq \beta^2+\beta+1$. The only previous algorithm to achieve something similar is the algorithm of \cite{MausT20}. As sketched in \cite{MausT20}, their can be adapted to compute list colorings with lists of size $|L_v|=\Theta(\beta^2\log\beta)$ in $O(\log^* n + \log^* C)$ rounds.

\begin{restatable}{theorem}{restatesecond}\label{thm:OLDC2}
  Let $G=(V,E)$ be a graph that is equipped with a proper $q$-vertex coloring and with an edge orientation. For every node $v\in V$, assume that $\beta_v$ is the outdegree of $v$. Further, assume that we are given an oriented list defective coloring instance for $G$ with color lists $L_v \subseteq \{1, \ldots, C\}$ and defect functions $d_v:L_v\to \mathbb{N}_0$ for all $v\in V$. Assume that
  \[
    \forall v \in V\,:\, \sum_{x\in L_v}(d_v(x)+1) \geq 3 \cdot \sqrt{C}\cdot\beta_v.
  \]
  Then, there is a deterministic  $O\big(\log^3C + \log^*q\big)$-round algorithm that solves the given oriented list defective coloring instance. The algorithm requires the nodes to exchange messages of $O(\log q + \log C)$ bits.
\end{restatable}

As a direct application, we can plug the above theorem into the framework of \cite{fraigniaud16local,Kuhn20,FK23} to obtain a new efficient $(\Delta+1)$-coloring for the \CONGEST model.

\begin{restatable}{theorem}{restatethird}\label{thm:CONGESTcoloring}
  Let $G=(V,E)$ be an $n$-node graph of maximum degree $\Delta$ and assume that all nodes $v\in V$ have color lists $L_v$ of size at least $|L_v|\geq \deg(v)+1$, consisting of colors from a color space of size $O(\Delta)$. Then the given list coloring instance can be solved in $O(\sqrt{\Delta}\cdot\log^4\Delta + \log^* n)$ deterministic rounds in the \CONGEST model.
\end{restatable}

We note that the above algorithm is by almost an $\Theta(\log^2\Delta)$-factor slower than the algorithm of \cite{FK23}. Further, the algorithm of \cite{FK23} also has a better dependency on the size of the color space. However, by replacing the oriented list defective coloring algorithm in \cite{FK23} with the algorithm of \Cref{thm:OLDC2}, our algorithm becomes simpler and computationally significantly more efficient than the algorithm of \cite{FK23} (which is the only existing $O(\sqrt{\Delta}\poly\log\Delta +\log^* n)$-time $(\Delta+1)$-coloring algorithm for the \CONGEST model). Furthermore, the algorithm of \cite{FK23} has the disadvantage that it only works in the \CONGEST model as long as $\Delta\leq \poly\log n$. For larger $\Delta$, one has to use alternative algorithm such as the algorithm of \cite{GhaffariKuhn21}, which has a round complexity of $O(\log^2\Delta\cdot \log n)$. Our algorithm from \Cref{thm:CONGESTcoloring} is a $\tilde{O}(\sqrt{\Delta})+O(\log^* n)$-round coloring algorithm in \CONGEST for the whole range of degrees.

\paragraph{\textbf{Coloring Graphs with Bounded Neighborhood Independence.}} We now summarize our results for graphs of bounded neighborhood independence. Before doing this, we introduce some notation that will be convenient for discussing and analyzing the technical results of this part of the paper. Recall that in an list defective or arbdefective coloring instance, every node $v$ receives a list $L_v$ of colors that come from a range of size $C$ (for simplicity, assume that $L_v\subseteq\set{1,\dots,C}$. In addition, there is a defect function $d_v:L_v\to \mathbb{N}_0$ for each node $v$. We define the following notion of \emph{slack}.

\begin{definition}[Slack]\label{def:slack}
  Consider a graph $G=(V,E)$ and a list defective or arbdefective coloring instance with lists $L_v$ and defect functions $d_v$ for all $v\in V$. We say that the instance has slack $S\geq 1$ if
  \[
    \sum_{x \in L_v} (d_v(x) + 1) > S \cdot deg(v).
  \]
\end{definition}

We will recursively solve the problems based on two main parameters, the slack $S$, and the size of the color space $C$. We use $P_D(S, C)$ to denote the family of list defective coloring instances, with slack at least $S$, and with a color space of size at most $C$. Similarly, we use $P_A(S,C)$ to denote the corresponding list arbdefective coloring problems. The round complexities of solving problems of the form $P_D(S, C)$ and $P_A(S, C)$, is denoted by $T_D(S, C)$ and $T_A(S, C)$, respectively.

Recall that for solving a problem in $P_D(S,C)$, we need to compute a coloring, where each node $v$ is colored with a color $x\in L_v$ such that the number of neighbors with color $x$ is at most $d_v(x)$. When solving a problem in $P_A(S,C)$, an algorithm in addition has to compute an orientation of the (monochromatic) edges and it must satisfy that the number of outneighbors of a node $v$ of color $x$ is at most $d_v(x)$. Above, we discussed that in graphs of bounded neighborhood independence $\theta$, an $O(d\theta)$-defective coloring can be computed by computing a $d$-arbdefective coloring. The next lemma generalizes this idea to the case of list defective and list arbdefective coloring. Instead of a single arbdefective coloring, in this more general case, we have to use $O(\log\Delta)$ consecutive carefully constructed instances of list arbdefective coloring.

\begin{restatable}{theorem}{restateforth}\label{thm:ArbToDef}
  Let $G=(V,E)$ be an $n$-node graph of maximum degree $\Delta$. For $S \geq 1$, there is a \CONGEST algorithm solving list defective coloring instances of the form $P_D(42 \cdot \theta \cdot \log \Delta \cdot S, C)$ that has a round complexity of
  \begin{align} \label{eq:defToArbdef}
      T_D(42 \cdot \theta \cdot \log \Delta \cdot S, C) \leq O(\log \Delta) \cdot T_A(S, C).
  \end{align} 
\end{restatable}

We can combine the above theorem with a recursive coloring framework that has been developed in particular in \cite{fraigniaud16local,Kuhn20,BalliuKO20,FK23} to obtain a fast $(\Delta+1)$-coloring algorithm and more generally an algorithm that solves arbitrary list arbdefective coloring instances with slack at least $1$. 

\begin{restatable}{theorem}{restatefive}\label{thm:mainNeighborhood}
  Let $G=(V,E)$ be an $n$-node graph of maximum degree $\Delta$ and neighborhood independence $\theta$. There is a deterministic \CONGEST algorithm solving an arbdefective coloring instance $P_A(1, O(\Delta))$ in 
  \begin{align*}
      T_A(1, O(\Delta)) \leq \min \left\{ (\theta \cdot \log \Delta)^{O(\log \log \Delta)}, O(\theta^2 \cdot \Delta^{1/4} \cdot \log^8 \Delta) \right\} + O(\log^* n)
  \end{align*} 
  communication rounds. 
\end{restatable}

A direct implication of \Cref{thm:mainNeighborhood} is that graphs with bounded neighborhood independence $\theta = O(1)$ can be properly colored in $O(\log^{\log \log \Delta} \Delta + \log^* n)$ communication rounds in \CONGEST with $\Delta+1$ colors. For instance, as discussed, line graphs of bounded rank hypergraphs form a family of graphs with bounded neighborhood independence. This property can be used to $2\Delta - 1$-color the edges of a graph in $O(\log^{\log \log \Delta} \Delta + \log^* n)$ rounds by simulating a coloring on the nodes of its line graph.

Another interesting implication of \Cref{thm:mainNeighborhood} is that in terms of $\Delta+1$-coloring in \CONGEST, this result can beat the $O(\sqrt{\Delta} \cdot \polylog \Delta + \log^* n)$ state-of-the-art complexity as of \cite{FK23} or \Cref{thm:CONGESTcoloring} for certain values of $\theta$. If $\theta = \tilde{O}(\Delta^{1/8})$ we get such a round complexity and if $\theta = \tilde{O}(\Delta^{1/8 - \eps})$ for some $\eps > 0,$ we even perform better.

\subsection{Organization of the paper}
\label{sec:organization}

The remainder of the paper is organized as follows. In \Cref{sec:model}, we introduce some additional definitions and notation conventions. In \Cref{sec:OLDC}, we provide the technical details for \Cref{thm:2sweepalgo,thm:OLDC2,thm:CONGESTcoloring} and in \Cref{sec:BoundedNeighborhood}, we prove \Cref{thm:ArbToDef,thm:mainNeighborhood}.


\section{Model and Preliminaries}
\label{sec:model}

\para{\textbf{Communication Model.}} In this paper we focus on the following distributed two distributed models, the \LOCAL model and the \CONGEST model~\cite{Peleg2000}. In both, the network is abstracted as an $n$-node graph $G=(V, E)$ in which each node is equipped with a unique $O(\log n)$-bit identifier. Communication between the nodes of $G$ happens in synchronous rounds. In every round, every node of $G$ can send a potentially different message to each of its neighbors, receive the messages from the neighbors and perform some arbitrary internal computation. The two models differ in the size they allow for messages. The \LOCAL model does not bound the message size at all, whereas in the \CONGEST model, messages must consist of at most $O(\log n)$ bits. 
Even if $G$ is a directed graph, we assume that communication can happen in both directions. All nodes start an algorithm at time $0$ and the time or round complexity of an algorithm is defined as the total number of rounds needed until all nodes terminate (i.e., output their color in a coloring problem).  Initially, the nodes do not know anything about the graph (except possibly some estimates on global parameters such as the number of nodes $n$ or the maximum degree $\Delta$) and at the end of a distributed algorithm, each node of $G$ must know its part of the output, e.g., its color when computing a vertex coloring. 

\para{\textbf{Mathematical Notation.}} Let $G=(V,E)$ be a graph. Throughout the paper, we use $\deg_G(v)$ to denote the degree of a node $v\in V$ in $G$ and by $\Delta(G)$ we denote the maximum of $2$ and the maximum degree. If $G$ is a directed graph, we further use $\beta_{v,G}$ to denote the outdegree of a node $v\in V$. More specifically, for convenience, we define $\beta_{v,G}$ as the maximum of $1$ and the outdegree of $v$, i.e., we also set $\beta_{v,G}=1$ if the outdegree of $v$ is $0$. The maximum outdegree $\beta_G$ of $G$ is defined as $\beta_G:=\max_{v\in V}\beta_{v,G}$. We further use $N_G(v)$ to denote the set of neighbors of a node $v$ and if $G$ is a directed graph, we use $\Nout_G(v)$ to denote the set of outneighbors of $v$. In all cases, if $G$ is clear from the context, we omit the subscript $G$. \par 

When discussing one of the list defective coloring problems on a graph $G=(V,E)$, we will typically assume that $\calC$ denotes the space of possible colors, and we use $L_v$ and $d_v$ for $v\in V$ to denote the color list and defect function of node $v \in V$. Throughout the paper, we will w.l.o.g.\ assume that $\calC \subseteq \mathbb{N}$ is a subset of the natural numbers. When clear from the context, we do not explicitly introduce this notation each time. \par

The \emph{neighborhood independence} $\theta$ of $G$ is the size of the largest independent set of the induced subgraphs $G[N_G(v)]$ for all nodes $v \in V$. \par

Further, we use $\log(x):=\log_2(x)$ and $\ln(x) :=\log_e(x)$.

\section{Oriented List Defective Coloring}
\label{sec:OLDC}

\subsection{The Two-Sweep Algorithm.} Before we dive into the details of \Cref{thm:2sweepalgo}, we introduce a simpler, yet less flexible variant of it. More specifically, we start with the special case of the algorithm of \Cref{thm:2sweepalgo} where $\eps = 0$ and we extend it to the more general case in \Cref{sec:fastersweep}. For that, assume that the given input graph $G$ comes together with an initial (proper) $q$-coloring and an edge orientation. Further, each node $v$ is equipped with a color list $L_v$ that is a subset of the known global color space $\calC$, whereas each color $x \in L_v$ allows some defect $d_v(x)$ among the outneighbors of $v$. Let $p\geq1$ be an integer parameter of our algorithm. Assume that for all nodes $v \in V$, the following inequality holds:

\begin{align}\label{eq:SlackInBaseCase}
    \sum_{x \in L_v} (d_v(x) + 1) > \max\left\{p, \frac{|L_v|}{p}  \right\} \cdot \beta_v.
\end{align}

We note that throughout the algorithm, nodes do not consider their in-neighbors to steer their decisions, or in other words, nodes at all times only \emph{look} towards their outneighbors. The algorithm proceeds in $2$ phases that we call Phase $I$ and Phase $II$. In both phases, we iterate through the $q$ initial colors, in ascending order in Phase $I$ and descending order in Phase $II$. The aim of Phase $I$ is that each node $v$ picks a sublist $S_v \subseteq L_v$ that contains at most $p$ colors that do not occur "often" within the color lists of $v$'s outneighbors. In Phase $II$, the nodes decide on some color $x \in S_v$. \par
Before we start with the details of Phase $I$, we give some additional notations. Let $v$ be a node with color $c \in \{1, \ldots, q \}$ from the initial $q$ coloring. Let $N_<(v)$ be the set of outneighbors of $v$ that have an initial color in $1, \ldots, c-1$ and let $N_>(v)$ be the outneighbors that have an initial color in $c+1, \ldots, q$. Note that by the definition of a proper coloring, there is no outneighbor that is also colored with color $c$ and thus $\beta_v = |N_<(v)| + |N_>(v)|$. The details of the algorithm appear in \Cref{alg:Sweep}. An illustration of the highlevel ideas of the two-phase algorithm appears in \Cref{fig:sweepidea}.

\begin{figure}[t]
    \centering
    \usetikzlibrary{arrows.meta}
\tikzset{every picture/.style={line width=0.75pt}} 

\begin{tikzpicture}[x=0.75pt,y=0.75pt,yscale=-1,xscale=1]

\draw   (136,95.5) .. controls (136,52.15) and (164.09,17) .. (198.75,17) .. controls (233.41,17) and (261.5,52.15) .. (261.5,95.5) .. controls (261.5,138.85) and (233.41,174) .. (198.75,174) .. controls (164.09,174) and (136,138.85) .. (136,95.5) -- cycle ;

\draw   (406,98.5) .. controls (406,55.15) and (434.09,20) .. (468.75,20) .. controls (503.41,20) and (531.5,55.15) .. (531.5,98.5) .. controls (531.5,141.85) and (503.41,177) .. (468.75,177) .. controls (434.09,177) and (406,141.85) .. (406,98.5) -- cycle ;

\draw [shorten >= 4pt, shorten <= 4pt, ->, >=Stealth, line width=0.80pt]   (332.5,93) -- (203.25,55.25) ;

\draw [shorten >= 4pt, shorten <= 4pt, ->, >=Stealth, line width=0.80pt]   (332.5,93) -- (155.25,77.25) ;

\draw [shorten >= 4pt, shorten <= 4pt, ->, >=Stealth, line width=0.80pt]   (332.5,93) -- (167.5,99) ;

\draw [shorten >= 4pt, shorten <= 4pt, ->, >=Stealth, line width=0.80pt]   (332.5,93) -- (171.25,123.25) ;

\draw [shorten >= 4pt, shorten <= 4pt, ->, >=Stealth, line width=0.80pt]   (332.5,93) -- (215.25,137.25) ;

\draw [shorten >= 4pt, shorten <= 4pt, ->, >=Stealth, line width=0.80pt]   (332.5,93) -- (473.25,58.25) ;

\draw [shorten >= 4pt, shorten <= 4pt, ->, >=Stealth, line width=0.80pt]   (332.5,93) -- (425.25,80.25) ;

\draw [shorten >= 4pt, shorten <= 4pt, ->, >=Stealth, line width=0.80pt]   (332.5,93) -- (468.75,93) ;

\draw [shorten >= 4pt, shorten <= 4pt, ->, >=Stealth, line width=0.80pt]   (332.5,93) -- (508.25,118.25) ;

\draw [shorten >= 4pt, shorten <= 4pt, ->, >=Stealth, line width=0.80pt]   (332.5,93) -- (441.25,126.25) ;

\draw [shorten >= 4pt, shorten <= 4pt, ->, >=Stealth, line width=0.80pt]   (332.5,93) -- (471.25,152.25) ;

\begin{scope}
    \draw  [fill={rgb, 255:red, 74; green, 144; blue, 226 }  ,fill opacity=1 ] (150,77.25) .. controls (150,74.35) and (152.35,72) .. (155.25,72) .. controls (158.15,72) and (160.5,74.35) .. (160.5,77.25) .. controls (160.5,80.15) and (158.15,82.5) .. (155.25,82.5) .. controls (152.35,82.5) and (150,80.15) .. (150,77.25) -- cycle ;

    \draw  [color={rgb, 255:red, 0; green, 0; blue, 0 }  ,draw opacity=1 ][fill={rgb, 255:red, 74; green, 144; blue, 226 }  ,fill opacity=1 ] (220.4,136.21) .. controls (220.97,139.05) and (219.13,141.82) .. (216.29,142.4) .. controls (213.45,142.97) and (210.68,141.13) .. (210.1,138.29) .. controls (209.53,135.45) and (211.37,132.68) .. (214.21,132.1) .. controls (217.05,131.53) and (219.82,133.37) .. (220.4,136.21) -- cycle ;

    \draw  [fill={rgb, 255:red, 74; green, 144; blue, 226 }  ,fill opacity=1 ] (198,55.25) .. controls (198,52.35) and (200.35,50) .. (203.25,50) .. controls (206.15,50) and (208.5,52.35) .. (208.5,55.25) .. controls (208.5,58.15) and (206.15,60.5) .. (203.25,60.5) .. controls (200.35,60.5) and (198,58.15) .. (198,55.25) -- cycle ;

    \draw  [fill={rgb, 255:red, 74; green, 144; blue, 226 }  ,fill opacity=1 ] (166,123.25) .. controls (166,120.35) and (168.35,118) .. (171.25,118) .. controls (174.15,118) and (176.5,120.35) .. (176.5,123.25) .. controls (176.5,126.15) and (174.15,128.5) .. (171.25,128.5) .. controls (168.35,128.5) and (166,126.15) .. (166,123.25) -- cycle ;

    \draw  [fill={rgb, 255:red, 227; green, 12; blue, 38 }  ,fill opacity=1 ] (325,93) .. controls (325,97.14) and (328.36,100.5) .. (332.5,100.5) .. controls (336.64,100.5) and (340,97.14) .. (340,93) .. controls (340,88.86) and (336.64,85.5) .. (332.5,85.5) .. controls (328.36,85.5) and (325,88.86) .. (325,93) -- cycle ;

    \draw  [fill={rgb, 255:red, 74; green, 144; blue, 226 }  ,fill opacity=1 ] (163.12,95.68) .. controls (161.28,98.1) and (161.76,101.55) .. (164.18,103.38) .. controls (166.6,105.22) and (170.05,104.74) .. (171.88,102.32) .. controls (173.72,99.9) and (173.24,96.45) .. (170.82,94.62) .. controls (168.4,92.78) and (164.95,93.26) .. (163.12,95.68) -- cycle ;

    \draw  [fill={rgb, 255:red, 126; green, 211; blue, 33 }  ,fill opacity=1 ] (420,80.25) .. controls (420,77.35) and (422.35,75) .. (425.25,75) .. controls (428.15,75) and (430.5,77.35) .. (430.5,80.25) .. controls (430.5,83.15) and (428.15,85.5) .. (425.25,85.5) .. controls (422.35,85.5) and (420,83.15) .. (420,80.25) -- cycle ;

    \draw  [fill={rgb, 255:red, 126; green, 211; blue, 33 }  ,fill opacity=1 ] (476.4,151.21) .. controls (476.97,154.05) and (475.13,156.82) .. (472.29,157.4) .. controls (469.45,157.97) and (466.68,156.13) .. (466.1,153.29) .. controls (465.53,150.45) and (467.37,147.68) .. (470.21,147.1) .. controls (473.05,146.53) and (475.82,148.37) .. (476.4,151.21) -- cycle ;

    \draw  [fill={rgb, 255:red, 126; green, 211; blue, 33 }  ,fill opacity=1 ] (513.4,117.21) .. controls (513.97,120.05) and (512.13,122.82) .. (509.29,123.4) .. controls (506.45,123.97) and (503.68,122.13) .. (503.1,119.29) .. controls (502.53,116.45) and (504.37,113.68) .. (507.21,113.1) .. controls (510.05,112.53) and (512.82,114.37) .. (513.4,117.21) -- cycle ;

    \draw  [fill={rgb, 255:red, 126; green, 211; blue, 33 }  ,fill opacity=1 ] (468,58.25) .. controls (468,55.35) and (470.35,53) .. (473.25,53) .. controls (476.15,53) and (478.5,55.35) .. (478.5,58.25) .. controls (478.5,61.15) and (476.15,63.5) .. (473.25,63.5) .. controls (470.35,63.5) and (468,61.15) .. (468,58.25) -- cycle ;

    \draw  [fill={rgb, 255:red, 126; green, 211; blue, 33 }  ,fill opacity=1 ] (468,58.25) .. controls (468,55.35) and (470.35,53) .. (473.25,53) .. controls (476.15,53) and (478.5,55.35) .. (478.5,58.25) .. controls (478.5,61.15) and (476.15,63.5) .. (473.25,63.5) .. controls (470.35,63.5) and (468,61.15) .. (468,58.25) -- cycle ;

    \draw  [fill={rgb, 255:red, 126; green, 211; blue, 33 }  ,fill opacity=1 ] (436,126.25) .. controls (436,123.35) and (438.35,121) .. (441.25,121) .. controls (444.15,121) and (446.5,123.35) .. (446.5,126.25) .. controls (446.5,129.15) and (444.15,131.5) .. (441.25,131.5) .. controls (438.35,131.5) and (436,129.15) .. (436,126.25) -- cycle ;

    \draw  [fill={rgb, 255:red, 126; green, 211; blue, 33 }  ,fill opacity=1 ] (463.25,93) .. controls (463.25,96.04) and (465.71,98.5) .. (468.75,98.5) .. controls (471.79,98.5) and (474.25,96.04) .. (474.25,93) .. controls (474.25,89.96) and (471.79,87.5) .. (468.75,87.5) .. controls (465.71,87.5) and (463.25,89.96) .. (463.25,93) -- cycle ;
\end{scope}

\draw [line width=2.25]    (102,211.5) -- (585,211.5) ;
\draw [shift={(589,211.5)}, rotate = 180] [color={rgb, 255:red, 0; green, 0; blue, 0 }  ][line width=2.25]    (17.49,-5.26) .. controls (11.12,-2.23) and (5.29,-0.48) .. (0,0) .. controls (5.29,0.48) and (11.12,2.23) .. (17.49,5.26)   ;

\draw (182,23) node [anchor=north west][inner sep=0.75pt]   [align=left] {$N_<(v)$};
\draw (451,27) node [anchor=north west][inner sep=0.75pt]   [align=left] {$N_>(v)$};
\draw (231,215) node [anchor=north west][inner sep=0.75pt]   [align=left] {increasing initial coloring};
\draw (341,71.75) node   [align=left] {\begin{minipage}[lt]{21.76pt}\setlength\topsep{0pt}
$v$
\end{minipage}};

\end{tikzpicture}
    \caption{\it The picture illustrates a node $v$ (in red) together with its outneighbors in $N_<(v)$ (in blue) and $N_>(v)$ (in green). Note that when node $v$ has to pick the set $S_v$ in Phase $I$, all blue nodes $u$ have already picked their subset $S_u$ of their color list and have sent it to $v$. In Phase $II$, $v$ has to output a final color from $S_v$. At this point, the green nodes ($N_<(v)$) have already picked their final colors. Node $v$ therefore has to pick its color based on knowing the setst $S_u$ for nodes in $N_<(v)$ and the final colors of nodes in $N_>(v)$.}
    \label{fig:sweepidea}
\end{figure}
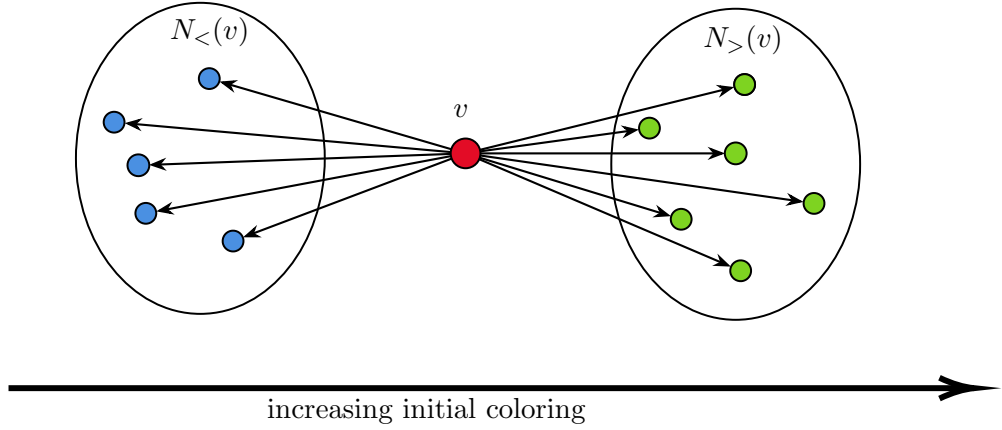

\begin{algorithm}
    \caption{Two-Sweep-Algorithm($q, p$)} \label{alg:Sweep}
    \begin{algorithmic}[1]
        \Require Proper $q$-coloring of the nodes. Integer $p \geq 1$. For each node $v \in V$, \Cref{eq:SlackInBaseCase} is true.
        \Ensure Oriented List Defective Coloring (OLDC) in time $O(q)$
        \State \textit{/** Phase I **/}
        \For{each node $v$ with initial color $c \in \{1, \ldots, q\}$} 
            \State For all $x \in L_v$, define $k_v(x)$ as the number of outneighbors $u$ of $v$ in $N_<(v)$ with $x \in S_u$. 
            \State Compute $S_v \subseteq L_v$ by sorting $L_v$ regarding $d_v(x) - k_v(x)$ in descending order and take the first $p$ colors.
        \EndFor
        \State \textit{/** Phase II **/}
        \For{each node $v$ with initial color $c \in \{q, \ldots, 1\}$} 
            \State For all $x \in L_v$, define $r_v(x)$ as the number of outneighbors of $N_>(v)$ that decided on color $x \in S_v$.
            \State Compute color $x_v \in S_v$, s.t.
            \begin{align*}
                k_v(x_v) + r_v(x_v) \leq d_v(x_v) 
            \end{align*}
            \State Decide on color $x_v$
        \EndFor
    \end{algorithmic}
\end{algorithm}

\paragraph{\textbf{Phase I}} In this phase we iterate over the initial $q$-coloring of the nodes (in ascending order) while all nodes with the same color pick a sublist $S_v$ (of size at most $p$) in parallel and send them to their neighbors. We call these nodes \emph{active}. We define $k_v(x)$ for color $x \in L_v$ as the number of outneighbors $u$ of $v$ in $N_<(v)$ with $x \in S_u$. Note that all nodes in $N_<(v)$ have already decided on it's sublist in this phase. Node $v$ will compute a set $S_v$ such that the following two conditions are true. 

\begin{align}
    |S_v| &\leq p  \label{eq:condition1a} \\
    |N_>(v)| + \sum_{x \in S_v} k_v(x) &< \sum_{x \in S_v} (d_v(x) + 1) \label{eq:condition1b}
\end{align}

The technical main statement of this phase is, that each node $v$ with sufficient large slack, i.e., \Cref{eq:SlackInBaseCase} holds, there exists a list $S_v$ that fulfills \Cref{eq:condition1a} and \Cref{eq:condition1b}. 

\begin{lemma}
    Let $v$ be an active node. There exists a set $S_v \subseteq L_v$ that fulfills requirements \Cref{eq:condition1a} and \Cref{eq:condition1b}.
\end{lemma}
\begin{proof}
    We start the proof with the easy case $|L_v| \leq p$, cause here we can simply define $S_v := L_v$. This choice clearly fulfills \Cref{eq:condition1a}. Making use of the fact $\sum_{x \in L_v} k_v(x) \leq \sum_{u \in N_<(v)} |S_u| \leq |N_<(v)| \cdot p$ we can also see \Cref{eq:condition1b} holds by the following lines.

    \begin{align*}
        |N_>(v)| + \sum_{x \in S_v} k_v(x) &\leq \beta_v - |N_<(v)| + |N_<(v)| \cdot p \\
        &= \beta_v + (p-1) \cdot |N_<(v)| \\
        &\leq p \cdot \beta_v \\
        &< \sum_{x \in S_v} (d_v(x) + 1)  
    \end{align*}

    Note that the last line follows from \Cref{eq:SlackInBaseCase}. We will now switch to the more interesting case $|L_v| > p$. For that, we take a look on all possible subsets of $L_v$ of size exactly $p$:

    \begin{align*}
        \sum_{S \in \binom{L_v}{p}} \left( |N_>(v)| + \sum_{x \in S} k_v(x) \right) 
        &= \binom{|L_v|}{p} \cdot (\beta_v - |N_<(v)|) + \binom{|L_v| - 1}{p - 1} \cdot \sum_{x \in L_v} k_v(x) \\
        &\leq \binom{|L_v|-1}{p-1} \cdot \left( \frac{|L_v|}{p} (\beta_v - |N_<(v)|) + |N_<(v)| \cdot p \right) \\
        &\leq \binom{|L_v|-1}{p-1} \cdot \left( \frac{|L_v|}{p} \beta_v + |N_<(v)| \cdot \max \left\{0, p - \frac{|L_v|}{p} \right\} \right)\\
        &\leq \binom{|L_v|-1}{p-1} \cdot \left( \beta_v \cdot \max \left\{ \frac{|L_v|}{p}, p \right\} \right)\\
        &< \binom{|L_v|-1}{p-1} \cdot \sum_{x \in L_v} (d_v(x)+1)
    \end{align*}

    In the second step we used that for any positive integer $k$ and arbitrary $n$, $\binom{n}{k} = \frac{n}{k}\binom{n-1}{k-1}$ holds. By pigeonhole principle there exists a $p$ sized subset of $L_v$ - let us call it $S_v$ - with the following property:

    \begin{align*}
        |N_>(v)| + \sum_{x \in S_v} k_v(x) &\leq \frac{\sum_{S \in \binom{L_v}{p}} \left( |N_>(v)| + \sum_{x \in S} k_v(x) \right)}{\binom{|L_v|}{p}} \\
        &< \frac{p}{|L_v|} \cdot \sum_{x \in L_v} (d_v(x)+1) < \sum_{x \in L_v} (d_v(x)+1)
    \end{align*}

    In the second step, we again used the identity on binomial coefficients that we also used above, and in the last step we used the condition of this second case $p < |L_v|$. Thus, $S_v$ fulfills \Cref{eq:condition1a} and \Cref{eq:condition1b} which concludes the proof. 
\end{proof}

\emph{Remark}: In \Cref{alg:Sweep} we compute $S_v$ by taking $p$ colors $x \in L_v$ that have the highest $d_v(x) - k_v(x)$ value among all colors in $L_v$. Clearly, such a set $S_v$ is the best possible out of all $p$ sized subsets and thus fulfills \Cref{eq:condition1b}.

\paragraph{\textbf{Phase II}} Before this phase starts, all nodes $v \in V$ have already decided on a sublist $S_v$. In Phase $II$ we iterate over the $q$ colors again, but in reverse direction this time. Whenever it is node $v$'s turn, $v$ will decide on a color from $S_v$ and forward it to its neighbors. Hence, the outneighbors in $N_<(v)$ have not decided on a final color yet, while the nodes in $N_>(v)$ already did. By $r_v(x)$ we denote the number of outneighbors of $N_>(v)$ that decided on color $x \in S_v$. We call node $v$ \emph{active in Phase $II$} in the iteration where $v$ has to decide on such a color. When $v$ is active, it can choose any color $x \in S_v$ if the following condition holds true 
\begin{align}\label{eq:condition2}
    k_v(x) + r_v(x) \leq d_v(x). 
\end{align}

We next show that such a color always exists and that at the end of Phase $II$ we have computed a list defective coloring. 

\begin{lemma}\label{lemm:corectnessOfSweep}
    For each node $v$ there is at least one color $x_v \in S_v$ that fulfills \Cref{eq:condition2} when $v$ became active in Phase $II$. Further, at most $d_v(x)$ outneighbors of $v$ will have decided on the same color after termination of Phase $II$.
\end{lemma}
\begin{proof}
    For this proof, note that $\sum_{x \in S_v} r_v(x)$ is at most the number of outneighbors in $N_>(v)$
    \begin{align}
        \sum_{x \in S_v} r_v(x) + k_v(x) \leq |N_>(v)| +  \sum_{x \in S_v} k_v(x) < \sum_{x \in S_v} (d_v(x) + 1) \label{eq:existenceOfColor}
    \end{align}
    In the second step we used \Cref{eq:condition1b}. By the pigeonhole principle there is some color in $l_v'$, let's call it $x_v$, such that $r_v(x_v) + k_v(x_v) \leq |N_>(v)| < d_v(x_v) + 1$. Since $r_v$, $k_v$ and $d_v$ map colors to non-negative integers, such an $x_v$ satisfies the statement of the lemma. To see that after termination of Phase $II$ at most $d_v(x_v)$ outneighbors have chosen $x_v$ as well, recognize that at most $r_v(x_v)$ neighbors in $N_>(v)$ have decided on the very same color and no more nodes from that set will be colored in future iterations, while from $N_<(v)$ every node that has $x_v$ in it's sublist (from Phase $I$) can cause a conflict in prospective iterations, however, these number is upper bounded by $k_v(x)$. 
\end{proof}

The following theorem concludes the base algorithm. 
\begin{lemma}
    Given an initial $q$ coloring and orientation of a graph $G$. The algorithm \Cref{alg:Sweep} solves an oriented list defective coloring instance parameter $1 \leq p$ in $O(q)$ communication rounds. 
    The algorithm requires the nodes to exchange a list of $p$ colors from $L_v$. \label{lemm:VanillaColoring}
\end{lemma}
\begin{proof}
    The correctness of the algorithm follows from \Cref{lemm:corectnessOfSweep}. The complexity of the algorithm is $O(q)$ since in both phases we have to iterate over all the $q$ colors whereas in Phase $I$ each node $v$ has to transmit it's reduced color list $S_v$ (that by \Cref{eq:condition1a} is of size at most $p$) to all neighbors while in Phase $II$ nodes simply have to transmit its final color. 
\end{proof}

This Lemma already proofs our main contribution \Cref{thm:2sweepalgo} for the $\eps = 0$ case. In the next section we generalize this result.

\subsection{Working With a Non-Proper Initial Coloring}\label{sec:fastersweep}
The weakness of the base algorithm is the dependence of the complexity on the $q$ coloring. We do not have sufficiently fast (proper) coloring algorithms to use here. However, we overcome this problem by using a defective coloring as of \Cref{lemm:defColorBlackBox} instead of a proper coloring. Edges that are monochromatic can be ignored in the following steps if we assume a higher slack as in \Cref{eq:SlackInSweep}.

\begin{lemma}[\cite{Kuhn2009, KawarabayashiS18}] \label{lemm:defColorBlackBox}
    For some $1 \geq \alpha > 0$, there exists a defective coloring algorithm in the \CONGEST model, that colors the nodes of a directed graph with $O(1/\alpha^2)$ many colors such that each node has at most $\alpha \cdot \beta_v$ outneighbors of the same color. This takes $O(\log^* q)$ communication rounds if the graph was already colored with an $O(q)$ coloring.
\end{lemma}

Let $\eps > 0$ be an integer and let $p$ be an integer as before, we will now dive into problem instance with the following specification for each node $v$:

\begin{align}\label{eq:SlackInSweep}
    \sum_{x \in L_v} (d_v(x) + 1) > \left( 1 + \eps \right) \cdot \max\left\{p, \frac{|L_v|}{p}  \right\} \cdot \beta_v
\end{align}

Assume we have a defective coloring given on $G$, such that every node $v$ has at most $\varepsilon/p \cdot \beta_v$ outneighbors of the same color. The first step is to get rid of monochromatic edges by this initial coloring. Thus, we define $G'$ as the subgraph of $G$ that contains all nodes from $G$ but only the edges where both endpoints are colored differently by the initial coloring. Note that $G'$ is a proper colored graph. The idea is to use \Cref{lemm:VanillaColoring} on $G'$, however, this would not solve the problem on $G$. For that, we reduce the colors defects before we apply \Cref{lemm:VanillaColoring} on $G'$. By this, we "save" some defect to get a valid defective coloring on $G$. On $G'$ we use a new defect function $d_v'(x)$ for all colors $x \in L_v$ and adjust the color list to $L_v' \subseteq L_v$ that consists of the colors with non-negative defect regarding $d_v'(x)$. The details are stated in the pseudocode of \Cref{alg:Sweep2}:

\begin{algorithm}
    \caption{Fast-Two-Sweep-Algorithm($q, p, \eps$)} \label{alg:Sweep2}
    \begin{algorithmic}[1]
        \Require Proper $q$-coloring of the nodes. Integer $p \geq 1$ and $\varepsilon > 0$. For each node $v \in V$ \Cref{eq:SlackInSweep} is true.
        \Ensure OLDC in time $O(\min \{q, (p/\eps)^2 + log^* q \})$
        \If{$q \leq \frac{p^2}{\varepsilon^2} + \log^* q$}
            \State Use Sweep-Algorithm($q$, $p$) 
        \EndIf
        \State Compute defective coloring $\Psi(v)$ using \Cref{lemm:defColorBlackBox} setting $\alpha := \varepsilon/p$. 
        \State Let $G'$ be the subgraph of $G$ removing monochromatic edges regarding $\Psi$. Further, let
        \begin{align*}
            d_v'(x) &:= d_v(x) - \left\lfloor \beta_v \cdot \frac{\varepsilon}{p}\right\rfloor \\
            L_v' &:= \{ x \in L_v \ | \ d_v'(x) \geq 0 \}
        \end{align*} 
        \State Solve the OLDC instance with defects $d_v'(x)$ and color list $L_v'$ on $G'$ using our Two-Sweep-Algorithm($O(p^2/\eps^2)$, $p$) (i.e., \Cref{alg:Sweep})
        \State Color node $v$ with the color computed in the previous line 
    \end{algorithmic}
\end{algorithm}

We will now show that \Cref{alg:Sweep2} solves oriented list defective coloring instances as defined in \Cref{thm:2sweepalgo}.

\restatefirst*
\begin{proof}
    First, assume that $\varepsilon > 0$, since the opposite is already proven in \Cref{lemm:VanillaColoring}. If  $q \leq \frac{p^2}{\varepsilon^2} + \log^* q$, the problem instance can also be solved using \Cref{lemm:VanillaColoring}. The runtime here is $O(q)$. If the reverse of this assumption is true, we compute a defective coloring using \Cref{lemm:defColorBlackBox} with the parameter $\alpha = \varepsilon/p$ in time $O(\log^* q)$. We will now show that the new problem instance on $G'$ (with updated defects and lists) fulfill \Cref{eq:SlackInBaseCase} and thus \Cref{alg:Sweep} can be applied.

    \begin{align*}
        \sum_{x \in L_v'} (d_v'(x) + 1) \geq \sum_{x \in L_v} (d_v(x) + 1) - \sum_{x \in L_v} \left\lfloor \beta_v \cdot \frac{\varepsilon}{p}\right\rfloor > \max\left\{p, \frac{|L_v|}{p} \right\} \beta_v
    \end{align*}

    From here assume all nodes got colored. It remains to show that this leads to a valid coloring in $G$. Since each node has at most $d'(x)$ outneighbors in $G'$ with the same color and at most $\beta_v \cdot \varepsilon/p$ additional outneighbors connected with monochromatic edges with respect to \Cref{lemm:defColorBlackBox}, there are at most $d'(x)+ \lfloor \beta_v \cdot \varepsilon/p \rfloor = d_v(x)$ outneighbors with the same color. The runtime here is linear in the number of colors we have in our defective coloring and thus is $O(p^2/\varepsilon^2)$. The statement of the lemma follows by combining both cases. 
\end{proof}

\paragraph{\textbf{Oriented List Defective Coloring in CONGEST.}} In this section we \emph{reduce} the message complexity of above's algorithm using the color space reduction technique as stated in \Cref{lemm:colorSpaceReductionBox}. Note that this lemma is a special case of the more general Theorem 3 in \cite{FK23}.

\begin{lemma}[Theorem 3 in \cite{FK23}] \label{lemm:colorSpaceReductionBox}
    Let $\kappa(\Lambda)$, $M(\Lambda)$ and $T(\Lambda)$ be functions of the maximum list size $\Lambda$. Assume we are given a deterministic algorithm $A$ solving oriented list defective coloring instance where for each node $v \in V$ we have 
    \begin{align*}
        \sum_{x \in L_v} (d_v(x) + 1) \geq \beta_v \cdot \kappa(\Lambda)
    \end{align*}
    with round complexity $T(\Lambda)$ and message complexity $M(\Lambda)$. Then, for any integer $\lambda \in \{1, \ldots, C \}$, there exists a deterministic algorithm $A'$ that solves oriented list defective coloring instance where for each node $v \in V$ we have 
    \begin{align*}
        \sum_{x \in L_v} (d_v(x) + 1) \geq \beta_v \cdot \kappa(\lambda)^{\lceil \log_\lambda(C) \rceil}
    \end{align*}   
    in time $O(T(\lambda) \cdot \lceil \log_\lambda C \rceil)$ and message complexity $M(\lambda)$.
\end{lemma}

Using it on our \emph{Sweep Algorithm} (as of \Cref{thm:OLDC2}) and choosing the right parameters we can reduce the message complexity to basically $O(\log C)$. This makes the algorithm applicable in \CONGEST if the size of the color space is at most $\poly n$. The details of the proofs appear in the following. 

\restatesecond*
\begin{proof}[Proof of \Cref{thm:OLDC2}]
    First, consider \Cref{thm:2sweepalgo} where we set the parameter $p := \left\lceil \sqrt{\Lambda} \right\rceil$ ($\Lambda$ is the maximum list size as in \Cref{lemm:colorSpaceReductionBox}) and $\eps := \frac{1}{3 \lceil \log_4 C \rceil}$. We get an algorithm that solves OLDC instance if 
    \begin{align*}
        \sum_{x \in L_v} (d_v(x) + 1) > \beta_v \cdot \left( 1+ \frac{1}{3 \lceil \log_4 C \rceil} \right) \cdot \left\lceil \sqrt{\Lambda } \right\rceil
    \end{align*}   
    in time $O(\Lambda \cdot \log^2 C + \log^* q)$ using messages of size at most $O(\log q + \sqrt{\Lambda} \cdot \log C)$ bits. Note the messages that has to be transmitted are the defective color that takes $O(\log q)$ bits, and the list of $\leq p$ colors whereas each color needs at most $O(\log C)$ bits. This algorithm with these parameters is used as parameter $A$ as of \Cref{lemm:colorSpaceReductionBox}. We will now apply this lemma with splitting parameter $\lambda = 4$. By that we get a new algorithm with initial condition  
    \begin{align} \label{eq:afterColorSpaceReduction}
        \sum_{x \in L_v} (d_v(x) + 1) > \beta_v \cdot \left( \left( 1+ \frac{1}{3 \cdot \lceil \log_4 C \rceil} \right) \cdot 2 \right)^{ \lceil log_4 C \rceil}
    \end{align}   
    on each node $v \in V$, that has a round complexity of $O(\log^3 C + \log^* q)$ and uses messages of size at most $O(\log q + \log C)$ bits. The statement of the lemma holds by some simplification of \Cref{eq:afterColorSpaceReduction}.
    \begin{align*}
        \left( \left( 1+ \frac{1}{3 \lceil \log_4 C \rceil} \right) \cdot 2 \right)^{ \lceil log_4 C \rceil} &\leq 2 \cdot \left( 1+ \frac{1}{3 \lceil \log_4 C \rceil} \right)^{\lceil \log_4 C \rceil} \cdot 2^{\log_4 C} \\
        &< 2 \cdot e^{1/3} \cdot \sqrt{C} < 3 \sqrt{C} 
    \end{align*}
\end{proof}
We are now ready to proof our final theorem about $degree+1$-list coloring in\CONGEST.

\restatethird*
\begin{proof}[Proof of \Cref{thm:CONGESTcoloring}]
    In this given coloring instance all defects are zero, i.e., we have a list defective coloring instance where for all nodes $v \in V$ it holds $\sum_{x \in L_v} d_v(x) + 1 = |L_v| \geq \deg(v) + 1$. Thus, we can apply Theorem 4 from \cite{FK23} where we plug in \Cref{thm:OLDC2} as list defective coloring algorithm.
    This results in an algorithm that computes the required $(deg+1)$-list coloring in time 
    \begin{align*}
        O(\sqrt{C} \cdot \log \Delta \cdot T + \log^* n )
    \end{align*}
    where $T$ is the runtime from \Cref{thm:OLDC2}. Since they start by computing an initial $O(\Delta^2)$ coloring, we can also use this as the $q$-coloring in \Cref{thm:OLDC2}. As we assume $C = O(\Delta)$ we get
    \begin{align*}
        T = O(\log^3 C + \log^* (\Delta^2)) = O(\log^3 \Delta).
    \end{align*}
    Plugin in that $T$ in the above's algorithms runtime ends the proof. 
\end{proof}


\section{List Coloring with Neighborhood Independence}
\label{sec:BoundedNeighborhood}

\subsection{From Defective to Arbdefective Coloring}\label{sec:BoundedNeighborhoodsSubalg}
The goal of this section is to solve a  list defective coloring instance by using a list arbdefective algorithm. The high-level idea is to use the greedy defective algorithm from the \emph{Bounded Neighborhood Independence} introduction in \Cref{para:IntroBoundedNeighborhoodIndependence}, that, sightly restated, is summarized in the following claim.  

\begin{claim} \label{claim:arbToDef}
    Let $G=(V, E)$ be a graph with neighborhood independence $\theta$ and a given$d$-arbdefective $C$-coloring. Then each node has at most $(2d+1) \cdot \theta$ neighbors of the same color, i.e., the arbdefective colors form a $(2d+1) \cdot \theta$-defective $C$-coloring.
\end{claim}

Intuitively, \Cref{claim:arbToDef} is what we want to use to solve a $P_D(\Omega(\theta), C)$ instance. However, this is not straightforward because different colors support different defects. To overcome this problem we restrict the color lists to colors of similar defects and then use the idea from the claim. To make sure that these shortened lists are of sufficient size we construct an iterative process such that for each node $v \in V$ there is an iteration where the reduced list $L_v$ fulfill all the properties for node $v$ to get colored. The details appear in the following.

\paragraph{\textbf{The Algorithm.}} We will now describe the algorithm that solves the problem $P_D(21 \cdot \theta \cdot (\lceil \log \Delta \rceil + 1) \cdot S, C)$ by using $P_A(S, C)$ as subroutine:

\begin{itemize}
    \item \textbf{Requirement:} Given $S \geq 1$, each node should satisfy the following condition:
    \begin{align}\label{eq:InitialSlackCondition}
        \sum_{x \in L_v} (d_v(x) + 1) > 21 \cdot \theta \cdot (\lceil \log \Delta \rceil + 1) \cdot S
    \end{align}
    \item Define new defect function for every color $x \in L_v$ for each node $v \in V$
    \begin{align} \label{eq:DefectDPrime}
        d_v'(x) := \left\lceil \frac{d_v(x) + 1}{7 \cdot \theta}\right\rceil - 1
    \end{align}
    Note that this new defect functions assigns non-negative integer defects to all the colors and has the following property:
    \begin{align}\label{eq:RemainingSlack}
        \sum_{x \in L_v} (d_v'(x) + 1) \geq \frac{1}{7 \theta} \sum_{x \in L_v} (d_v(x)+1) > 3 \cdot (\lceil \log \Delta \rceil+1) \cdot S \cdot \deg(v).
    \end{align}

    \item Iterate through iterations $i = \lceil \log \Delta \rceil, \ldots, 0$ whereas we denote the number of neighbors of $v$ that are colored with color $x$ in iterations $\lceil \log \Delta \rceil,...,i+1$ by $a_v(x, i)$. The number of neighbors of $v$ that are colored before iteration $i$ is $\widetilde{\deg}(v, i) \leq \sum_{x \in L_v} a_v(x, i)$. Furthermore, we define a round specific defect by $d_i := 2^i - 1$. Based on this $d_i$, each node $v$ furthermore restricts its list of 'allowed' colors in iteration $i$ to $L_{v, i} \subseteq L_v$. We say $x \in L_{v, i}$ iff $x \not \in L_{v, j}$ for all $j > i$ and
    \begin{align}\label{eq:ColorListinIteration}
         d_v'(x) - a_v(x, i) \geq d_i 
    \end{align}

    \item Let $H_i$ be the subgraph of $G$ of nodes $v$ that are still uncolored in iteration $i$ where 
    \begin{align}\label{eq:WillBeColored}
        \sum_{x \in L_{v, i}} (d_i + 1) > S \cdot (\deg(v) - \widetilde{\deg}(v, i)).
    \end{align}
    The algorithm colors all the nodes of $H_i$ using $P_A(S, C)$ (as node $v$ has degree at most $deg(v) - \widetilde{deg}(v, i)$ within $H_i$ we can guarantee enough slack). Edges will always be oriented towards already colored nodes such that no conflict arises when nodes will be colored in later iterations. 
\end{itemize}

We will now analyze this algorithm. Note that there are two crucial statements to show. First, we have to show that for every node $v \in V$, the exists an iteration $i \in \{ \lceil \log \Delta \rceil, \ldots, 0 \}$ where $v$ will get colored (i.e., \Cref{eq:WillBeColored} holds for $v$) and second, that the coloring of $P_A(S, C)$ that will be used to color nodes in $H_i$ fulfills requirement of the original list defective coloring. The formal statements are written in \Cref{lemm:WillBeColored} and \Cref{lemm:Neighbors}. \par

\begin{lemma}\label{lemm:WillBeColored}
    For every node $v \in V$, there is an iteration $i_v \in \{0, \ldots, \lceil \log \Delta \rceil\}$ where node $v$ will be colored.
\end{lemma}
\begin{proof}[Proof of \Cref{lemm:WillBeColored}]
    We first bound to the number of colors of $v$ that are not considered in any iteration, i.e., let $L_v^* \subseteq L_v$ be the colors of $L_v$ that are never element of $L_{v, i}$ for any iteration $i$. By \Cref{eq:ColorListinIteration} we have for every $x \in L_v^*$ and every iteration $i$ that $d_v'(x) + 1 \leq d_i + a_v(x, i)$. Since this holds in particular for the last iteration $i = 0$, we have $d_v'(x) + 1 \leq a_v(x, 0)$ and therefore $\sum_{x \in L_v^*}(d_v'(x)+1) \leq \widetilde{\deg}(v, 0) \leq \deg(v)$. 

    \begin{align*}
        \sum_{i=0}^{\lceil \log \Delta \rceil} \sum_{x \in L_{v, i}} (d_v'(x) + 1) &= \sum_{x \in L_v} (d_v'(x) + 1) - \sum_{x \in L_v^*} (d_v'(x) + 1) \\
        &> \left( 3 \cdot (\lceil \log \Delta \rceil + 1) \cdot S - 1 \right) \cdot \deg(v) \\
        &\geq 2  S \cdot (\lceil \log \Delta \rceil+1) \cdot \deg(v)
    \end{align*}

    By the pigeonhole principle it follows that for each node $v$ there is an iteration $i_v$ s.t. 
    \begin{align} \label{eq:pigeonDefects}
        \sum_{x \in L_{v, i_v}} (d_v'(x) + 1) > \frac{2 S (\lceil \log \Delta  \rceil + 1)}{\lceil \log \Delta \rceil  +1} \cdot deg(v)  \geq 2 \cdot S \cdot \deg(v)
    \end{align}

    Before we proceed, we show that for a color $x \in L_{v, i}$ it always holds $d_v'(x) - a_v(x, i) \leq 2 \cdot d_i$. For the first iteration $i = \lceil \log \Delta \rceil$ this is true since by $a_v(x, \lceil \log \Delta \rceil) = 0$ and $d_v(x) \leq \Delta$ (cause otherwise colors with such high defects could easily be picked by the nodes without taking care of neighbors) we have 
    \begin{align*}
        d_v'(x) - a_v(x, \lceil \log \Delta \rceil) \leq d_v(x) \leq \Delta \leq 2 \cdot (2^{\lceil \log \Delta \rceil } - 1) \leq 2 \cdot d_{\lceil \log \Delta \rceil}.
    \end{align*}
    Now assume $i < \lceil \log \Delta \rceil$. Since $x$ was not put into $L_{v, i+1}$ it follows
    \begin{align*}
        2 \cdot d_i \geq d_{i+1} - 1 > d_v'(x) - a_v(x, i+1) - 1 \geq d_v'(x) - a_v(x, i) - 1
    \end{align*}
    and thus, the claim follows by the fact that $d_i$ is an integer.

    Let's now proof the lemmas' statement by contradiction, i.e., assume there is a node $v$ such that there is no iteration $i$ where node $v$ is element of such a subgraph $H_i$ (by \Cref{eq:WillBeColored}) and thus never gets colored. This would imply that also for iteration $i_v$ (from \Cref{eq:pigeonDefects}) we have $\sum_{x \in L_{v, i_v}} (d_{i_v}(x) + 1) \leq S \cdot (\deg(v) - \widetilde{\deg}(v, i_v))$, and thus we get a contradiction to \Cref{eq:pigeonDefects} by the following lines

    \begin{align*}
        \sum_{x \in L_{v, i_v}} (d_v'(x) + 1) &\leq  \sum_{x \in L_{v, i_v}} (2d_{i_v} + 1 + a_v(x, i_v)) \\
        &\leq 2 \cdot S \cdot (\deg(v) - \widetilde{\deg}(v, i_v)) + \widetilde{\deg}(v, i_v) \\
        &\leq 2 \cdot S \cdot \deg(v).
    \end{align*}
\end{proof}

\begin{lemma} \label{lemm:Neighbors}
    At the end of the algorithm for every node $v$ that is colored with color $x \in L_v$, there are at most $X_v < \max \{1, 7 \theta \cdot d_v'(x) \}$ neighbors of $v$ with the same color.
\end{lemma}
\begin{proof}[Proof of \Cref{lemm:Neighbors}]
    W.l.o.g., let $v$ be colored in iteration $i$. To bound the number of neighbors, we go through 3 steps. We have to bound the number of neighbors of the same color in iterations $\lceil \log \Delta \rceil, ..., i+1$, the number of neighbors colored in the same iteration $i$ and the number of neighbors colored in future iterations $i-1, \ldots, 0$.

    \begin{enumerate}
        \item The number of neighbors colored before iteration $i$ is $a_v(x, i)$.
        \item By \Cref{claim:arbToDef}, at most $(2d_i+1) \cdot \theta$ neighbors are colored with $x$ in iteration $i$. As a remark, note that in the last iteration $d_{0} = 0$. Thus, if a node $v$ gets colored in the last iteration, no outneighbor will be colored with the same color and hence also no inneighbor will be colored with the same color in this last iteration.
        \item Let $G_{x, j}(v)$ be subgraph of $H_j$ that includes all nodes from $H_j$ that are colored with color $x$ and additionally includes $v$. Let's observe $G_{x, i-1}(v)$, where $v$ was colored with $x$ in iteration $i$. Note that there is an orientation on every edge, that in particular orients all edges towards $v$. Thus, the max. outdegree here is $d_{i-1}$ and hence by the same argument as in \Cref{claim:arbToDef} each node in $G_{x, j}(v)$ has at most $(2d_{i-1} + 1) \cdot \theta$ (undirected) neighbors of the color. From that we can conclude that node $v$ will get at most $(2d_{i-1} + 1) \cdot \theta$ additional neighbors of the same color in iteration $i-1$. Note that his generalized for all iterations $j < i$. 
    \end{enumerate}

    Let $X_v$ be the number of neighbors of $v$ with the same color. We will bound $X_v$ considering two different cases. \\
    \textbf{Case1: $v$ will be colored in iteration $i := 0$}: \\ 
    By above's observation we have $X_v \leq a_v(x, 0) \leq d_v'(x) - d_{0} = d_v'(x)$. \\
    \textbf{Case2: $v$ will be colored in iteration $i > 0$}: \\
    By aboves observations we have $X_v \leq a_v(x, i) + \sum_{j = 0}^{i} (2d_{j} + 1) \cdot \theta$. But first, let's take a look on some simplifications. Note that in this case $d_j \geq 1$ for all $j > 0$. Further, for every $j \leq i$ we can bound $d_j \leq \frac{d_i}{2^{i-j}}$. This bound can inductively be proven over $j$. For $j = i$ it trivially holds. If it holds for some $j$, then it also holds for $j-1$, because 
    \begin{align*}
        d_{j-1} = \frac{1}{2} \cdot 2^j - 1 \leq \frac{1}{2} (2^j - 1) = \frac{1}{2} d_j \leq \frac{d_i}{2^{i - (j - 1)}} .
    \end{align*}
    We can now apply the simplifications:
    \begin{align}
        X_v &\leq a_v(x, i) + \sum_{j = 0}^{i} (2d_{j} + 1) \cdot \theta \notag \\
        &\leq a_v(x, i) + \theta + \sum_{j = 1}^{i} 3 \cdot d_j \cdot \theta \notag \\
        &\leq a_v(x, i) + \theta + 3 \theta \cdot d_i \sum_{j = 1}^{i} \frac{1}{2^{i-j}} \notag \\
        &< a_v(x, i) + \theta + 3 \theta \cdot d_i \sum_{j = 0}^{\infty} \frac{1}{2^{j}} \notag \\
        &\leq a_v(x, i) + \theta + 6 \theta (d_v'(x) - a_v(x, i))  \notag \\
        &\leq 7 \theta \cdot d_v'(x) \label{eq:neighbors}
    \end{align}
    Combining the two cases lead to $X_v < \max\{d_v'(x) + 1, 7 \theta \cdot d_v'(x)\}$. Because $\theta \geq 1$, the first argument of the max is only larger when $d_v'(x) = 0$. Thus, the lemma is proven. 
\end{proof}

\restateforth*
\begin{proof}
    First, note that the slack as stated in the Theorem statement implies the slack as we need for the algorithm defined in \Cref{eq:InitialSlackCondition}. By \Cref{lemm:WillBeColored} each node will be colored in some iteration. So fix some node $v$ that got colored by color $x$. Combining \Cref{eq:DefectDPrime} and \Cref{lemm:Neighbors} we derive that  
    \begin{align*}
        X_v \leq \max \{ 1, 7 \theta \cdot d_v'(x) \} - 1 \leq \max \left\{0, 7 \theta \cdot \left(\frac{d_v(x) + 1}{7\theta} \right) - 1  \right\} = \max \{0, d_v(x) \} = d_v(x).
    \end{align*}
    We have now proven that the algorithm computes a valid list defective coloring. Each of the $\lceil \log \Delta \rceil + 1$ iterations takes $T_A(S, C)$ rounds. The only communication between two arbitrary nodes $u$ and $v$ is the information about the chosen color. Sending such a color takes $O(\log C)$ bits and is thus a valid \CONGEST message.    
\end{proof}

\subsection{Slack Reduction} \label{sec:SlackReduction}
The goal of this section is to solve a list arbdefective coloring instance with low slack by using a arbdefective coloring instance with higher slack by paying some additional communication rounds (with dependence on the new slack value). For our first result, assume we have to solve a $P_A(2, C)$ instance. In \Cref{sec:modifiedSlackReduction} we give a modified version that can also handle a smaller initial slack, this however comes with additional cost on the round complexity. \par

Assume that $G$ is already colored with some defective coloring as of \Cref{lemm:defColorBlackBox}, i.e., for some $\varepsilon < 1$ we assume a $K := O(1/\varepsilon^2)$ coloring s.t. at most $\varepsilon \cdot deg_G(v)$ neighbors have the same color as $v$, for all $v \in V$. To get a slack of $\mu > 1$, we will later see that it is sufficient to set $\varepsilon := 1/\mu$.\par
Let $G_i$ for $i \in \{1, \ldots, K\}$ be the subgraph of $G$ colored with color $i$ from this defective coloring. The algorithm starts by \emph{sequentially} iterate through the partitions $G_1, G_2, \ldots, G_K$ and color all the nodes within them. Whenever we are done coloring a subgraph $G_i$, we will orient the edges in all subsequent $G_j$ with $j > i$ towards the colored nodes in $G_i$ and adjust the color list and the defect values in the very same way, i.e., we define $L_v'$, $d_v'$, $a_v$ and $\widetilde{\deg}(v)$ as follows: Let $v$ be a node in $G_j$, we say $a_v(x)$ are the number of outneighbors with color $x$ in $G_1, \ldots, G_{j-1}$, by $\widetilde{\deg}(v) := \sum_{x \in L_v} a_v(x)$ we denote the already colored neighbors, $d_v'(x) := d_v(x) - a_v(x)$ is the new defect function and $L_v' \subseteq L_v$ contains all the colors $x \in L_v$ where $d_v(x) - a_v(x) \geq 0$.

\begin{lemma}\label{lemma:SlackReduction2}
    Let $G=(V,E)$ be a graph that is equipped with a proper $q$-vertex coloring. For $\mu \geq 1$, there exists a \CONGEST algorithm that solves $P_A(2, C)$ instances on $G$ in time
        \begin{align} \label{eq:slackGeneration2}
        T_A(2, C) \leq O(\mu^2) \cdot T_A\left(\mu, C \right) + O(\log^* q).
    \end{align}
\end{lemma}
\begin{proof}
    Let $v$ be a node of $G_j$ and let the nodes in $G_1, \ldots, G_{j-1}$ already be colored. For $v$ we have 
\begin{align*}
    \sum_{x \in L_v'}(d_v'(x) + 1) &\geq \sum_{x \in L_v} (d_v(x) + 1) - \sum_{x \in L_v} a_v(x) \\
    &> 2 \deg(v) - \widetilde{\deg}(v) \geq \deg(v) \geq \mu \cdot \deg_{G_i}(v)
\end{align*}
where the last line follows from the fact that $\deg_{G_i}(v) \leq \frac{1}{\mu} \cdot \deg(v)$ when we choose $\varepsilon = 1/\mu$. Since we have to go through all the $K = O(1/\varepsilon^2) = O(\mu^2)$ in sequential fashion, and the max. degree in a subgraph $G_i$ is $\frac{1}{\mu} \cdot \deg(v) \leq \frac{\Delta}{\mu}$, the round complexity of this procedure is $O(\mu^2) \cdot T_A(\mu, C )$. Note that the computation of the defective coloring by \Cref{lemm:defColorBlackBox} costs additional $O(\log^* q)$ rounds.
\end{proof}

In \Cref{sec:modifiedSlackReduction} we give a modified version of this algorithm that can also handle instances of the form $P_A(1, C)$. However, this modification comes with an additional $\log \Delta$ factor in the runtime. The details appear in \Cref{lemma:SlackReduction1}.

\subsection{Color Space Reduction}\label{sec:ColorSpaceReduction}
The idea of the color space reduction technique is to partition the color space $\calC$ into equally sized subspaces $\calC_1, \calC_2, \ldots, \calC_p$ and assign each node to exactly one subspace. Let node $v$ be assigned to color subspace $\calC_i$, then $v$'s color list is reduced to $L_{v, i} := L_v \cap \calC_i$. However, $v$ now only has to compete with neighbors that were assigned the same color subspace $\calC_i$. The choices of color subspaces by the nodes can itself be phrased as a list defective coloring instance for a color space of size $p$ and thus also with lists of size at most $p$. \par
Let $p \in \{1, \ldots, C\}$ be the number of subspaces of $\calC$. W.l.o.g., assume that $C$ is a multiple of $p$ as otherwise, we can just add some dummy colors to $\calC$. Thus, each subspace $\calC_i$ is of size exactly $C/p$ (and since we assume $\calC$ and $p$ are public, each node can construct the partition $\calC = \calC_1 \cup \ldots \cup \calC_p$). 
In the proof of the following lemma we show how nodes decide on such a subspace $\calC_1$.

\begin{lemma}\label{lemm:colorSpaceReduction}
    Let $1 \leq \sigma \leq S$ and $p \in \{1, \ldots, C\}$, then there is a algorithm solving $P_A(S, C)$ instances in time
    \begin{align} 
        T_A(S, C) \leq T_D\left(\sigma, p \right) + T_A\left(\frac{S}{\sigma}, \left\lceil \frac{C}{p} \right\rceil \right)
    \end{align}
\end{lemma}
\begin{proof}[Proof of \Cref{lemm:colorSpaceReduction}]
    First we define a new defect function and color list for each \emph{potential} color space $\calC_i$:
    \begin{align}
        L_{v, i} &:= L_v \cap \calC_i \label{eq:subSpacelist} \\
        d_{v, i} &:=  \left\lfloor \sigma \cdot deg(v) \cdot \frac{\sum_{x \in L_{v, i}}(d_v(x) + 1)}{\sum_{x \in L_{v}}(d_v(x) + 1)} \right\rfloor \label{eq:subsubSpaceDefect}
    \end{align}

    We start by showing that the defective coloring problem of choosing a 'color' $1 \leq i \leq p$ that has defect at most $d_{v, i}$ can be seen as a problem instance of $P_D(\Delta, \sigma, p)$, i.e.,

    \begin{align*}
        \sum_{i=1}^p (d_{v, i} + 1) > \sigma \cdot deg(v) \cdot \frac{\sum_{i=1}^p \sum_{x \in L_v \cap \calC_i}(d_v(x) + 1)}{\sum_{x \in L_{v}}(d_v(x) + 1)} =  \sigma \cdot \deg(v).
    \end{align*}

    We proceed by showing that the arbdefective coloring instance among the nodes that have picked the same sub space, is of the form $P_A(S/\sigma, \lceil C/p \rceil)$. Let $G_i$ be the subgraph of $G$ that exclusively contains the nodes that decided on sub space $\calC_i$. The degree of a node $v$ in $G_i$ is by construction at most $\deg_{G_i}(v_i) := \min \{d_{v, i}, \deg(v) \}$. The valid colors for node $v \in V_i$ are from $L_{v, i} \subseteq \calC_i$. Note that each subspace contains at most $\lceil C/p \rceil$ colors from $\calC$. The list arbdefective coloring instance of the nodes in $V_i$ has slack at least $S/\sigma$:

    \begin{align*}
        \sum_{x \in L_{v, i}} (d_v(x) + 1) \geq \frac{d_{v, i} \cdot \sum_{x \in L_v} (d_v(x) + 1)}{\sigma \cdot deg(v)} > \frac{S}{\sigma} \cdot  d_{v, i} \geq \frac{S}{\sigma} \cdot deg_{G_i}(v_i)
    \end{align*}

    Thus, solving $P_A(S, C)$ takes time $T_D\left(\sigma, p \right) + T_A\left(\frac{S}{\sigma}, \left\lceil \frac{C}{p} \right\rceil \right)$. 
\end{proof}

The downside of \Cref{lemm:colorSpaceReduction} is that we need an algorithm for list defective coloring instances. However, by \Cref{thm:ArbToDef} we know how to solve a list defective coloring instance using arbdefective instances.
The following lemma combines these thoughts and simplifies the usage for our purpose. 

\begin{lemma} \label{lemm:combiningArbdef}
    There is an algorithm that solves the list arbdefective coloring instance $P_A(84 \cdot \theta \cdot \log \Delta, C)$ on graphs with max. degree $\Delta$ in time
    \begin{align*}
        T_A(84 \cdot \theta \cdot \log \Delta, C) \leq O(\log \Delta) \cdot T_A(2, \lceil \sqrt{C}\rceil).
    \end{align*}
\end{lemma}
\begin{proof}
    We first apply \Cref{lemm:colorSpaceReduction} with parameters $p = \lceil \sqrt{C} \rceil$ and $\sigma = 42 \cdot \theta \cdot \log \Delta$ and then use \Cref{thm:ArbToDef}.
    \begin{align*}
        T_A(2\sigma, C) &\leq T_D\left(\sigma, \lceil \sqrt{C} \rceil \right) + T_A\left(\frac{2\sigma}{\sigma}, \left\lceil \frac{C}{\lceil \sqrt{C} \rceil} \right\rceil \right) \\
        &\leq O(\log \Delta) \cdot T_A(2, \lceil  \sqrt{C} \rceil) + T_A(2, \lceil \sqrt{C} \rceil ) \\
        &\leq O(\log \Delta) \cdot T_A(2, \lceil \sqrt{C} \rceil)
    \end{align*}
\end{proof}

\subsection{Putting Pieces Together}\label{sec:MainSection}
What we miss from the above's sections is some base case, i.e., some algorithm that solves a given problem instance without a recursive call. However, our Two-Sweep Algorithm from section \Cref{sec:OLDC} does the trick. To proof our main contribution \Cref{thm:mainNeighborhood} we can now put all these pieces from  \Cref{sec:BoundedNeighborhoodsSubalg}, \Cref{sec:SlackReduction} and \Cref{sec:ColorSpaceReduction} together. 

\restatefive*
\begin{proof}[Proof of \Cref{thm:mainNeighborhood}]
    In the following proof assume we are given a proper $O(\Delta^2)$ coloring on $G$. Note that we can easily get rid this assumption by using Linial's classic $O(\log^* n)$ round coloring algorithm (\cite{Linial1987}).  
    Assume the problem instance would be of form $P_A(2, C)$. We will lift this assumption later by \Cref{lemma:SlackReduction1}. Further, let $S := 84 \cdot \theta \cdot \log \Delta$ be the slack value from \Cref{lemm:combiningArbdef}.

    \begin{align*}
        T_A(2, C) &\leq O(S^2) \cdot T_A \left(S, C \right) + O(\log^* \Delta^2) \tag{\Cref{lemma:SlackReduction2} with $\mu := S$} \\
        &\leq O(S^2 \cdot \log \Delta) \cdot T_A \left(2, \lceil \sqrt{C} \rceil \right) \tag{\Cref{lemm:combiningArbdef}}
    \end{align*}

    Recursively using that procedure, say for $i \geq 1$ many steps, we get 
    \begin{align*}
        T_A(2, C) &\leq O(S^{2i} \cdot \log^{i} \Delta) \cdot T_A\left(2, C^{1/2^i} + i  \right)
    \end{align*}
    From here we choose $i$ in one of the two following ways, let's start with $i=1$ and apply \Cref{thm:CONGESTcoloring}. Note that by the proof of \Cref{thm:CONGESTcoloring} we can compute a instance of $P_A(1, C)$ in $T_A(1, C) \leq O(\sqrt{C} \cdot \log^3 C \cdot \log \Delta)$ communication rounds if the graph has some given initial coloring.
    \begin{align}\label{eq:ChoosingSmalli}
        T_A(2, C) &\leq O(\theta^2 \cdot \log^{3} \Delta) \cdot O(\sqrt{C^{1/2}} \cdot \log^4 \Delta) \leq C^{1/4} \cdot \theta^2 \cdot \log^4 \Delta \cdot \log^3 C
    \end{align}
    Now, we pick $i = \log \log C$ and make use of \Cref{thm:CONGESTcoloring} again. Note that by this choice of $i$, we have of $C^{1/2^i} = C^{1/(\log C)} = 2$.
    \begin{align}\label{eq:ChoosingLargei}
        T_A(2, C) &\leq O\left((84 \cdot \theta)^{2 \cdot \log \log C} \cdot \log^{3 \cdot \log \log C} \Delta \right) \cdot O\left(\sqrt{\log\log C} \cdot \log^3 C \cdot \log \Delta \right) 
    \end{align}
    The statement of the lemma follows by using the better of \Cref{eq:ChoosingSmalli} and \Cref{eq:ChoosingLargei} after an initial slack reduction of \Cref{lemma:SlackReduction1} with $\mu=2$.
    \begin{align*}
        T_A(1, O(\Delta)) &\leq O(\log \Delta) \cdot T_A(2, O(\Delta)) \\
        &\leq \min \left\{(\theta^2 \cdot \log \Delta)^{O(\log \log \Delta)}, O(\Delta^{1/4} \cdot \theta^2 \cdot \log^8 \Delta) \right\}
    \end{align*}
    As described in the beginning, we additionally have to spend $O(\log^* n)$ rounds to compute an initial $O(\Delta^2)$.
\end{proof}


\newcommand{\etalchar}[1]{$^{#1}$}

\appendix

\section{Modified Slack Reduction}
\label{sec:modifiedSlackReduction}

In contrast to our result in \Cref{sec:SlackReduction}, we consider a problem instance of the form $P_A(\Delta, 1, C)$ in this section. This modified algorithm proceeds in recursive steps whereas in every recursive step, a node $v$ will either be colored itself, or $v$ has at most $\Delta/2$ uncolored neighbors afterwards. From this, we construct the subgraph $G'$ of $G$ by deleting all colored nodes from $G$ (and thus $G'$ has max. degree at most $\Delta/2$). Before executing the next recursion step on $G'$, we orient the edges in the following fashion: Edges between colored and uncolored nodes will always be oriented towards the colored nodes. This ensures that already colored nodes will not violate its defect in later stages. Edges between colored nodes can be oriented arbitrarily. Edges between uncolored nodes lie completely in $G'$ and will therefore be oriented in the future. \par

\begin{lemma}\label{lemma:SlackReduction1}
    Let $G=(V,E)$ be a graph with max. degree $\Delta$ that is initially equipped with a proper $q$-vertex coloring. For $\mu \geq 1$, there exists a \CONGEST algorithm that solves $P_A(1, C)$ on $G$ in time
    \begin{align} 
         T_A(1, C) \leq O(\mu^2 \cdot \log \Delta) \cdot T_A\left(\mu, C \right) + O(\log^* q)
    \end{align}
\end{lemma}
\begin{proof}
    Here we go through a similar procedure as in \Cref{lemma:SlackReduction2}, i.e., we construct the subgraphs $G_i$ for $i \in \{1, \ldots, O(1/\eps^2)\}$ using \Cref{lemm:defColorBlackBox} but decreasing $\eps$ slightly, i.e., $\eps := 1/(2\mu)$. However, here we need to \textit{artificially} increase the slack. Let $H_i$ be the subgraph of $G_i$ that contains only nodes that have at most $\Delta/2$ colored neighbors in $G$, i.e.,  $\widetilde{\deg}_G(v) \leq \Delta/2$ colored neighbors in $G$. The next line proofs that nodes $v$ in $H_i$ have enough slack to be colored.
    \begin{align*}
        \sum_{x \in L_v'} (d_v'(x) + 1) &\geq \sum_{x \in L_v} (d_v(x) + 1) - \sum_{x \in L_v} a_v(x) \\
        &> \deg_G(v) - \widetilde{\deg}_G(v) \\
        &\geq \Delta - \Delta/2 \\
        &= \Delta/2 \\
        &\geq \mu \cdot \deg_{G_i}(v)
    \end{align*}
    The last inequality follows by the choice of $\varepsilon := \frac{1}{2 \mu}$, that by \Cref{lemm:defColorBlackBox} implies $deg_{G_i}(v) \leq \varepsilon \cdot deg_G(v) \leq \frac{1}{2\mu} \cdot \Delta$. Note that all nodes in $H_i$ will be colored in time $T_A\left(\mu, C \right)$. \par
    After our algorithm finished this coloring procedure for the nodes in $H_i$ for all $i \in \{1, \ldots, O(1/\eps^2)\}$ every node has $> \Delta(G)/2$ colored neighbors. We will now construct a new graph from $G'$ that contains only the uncolored nodes of $G$. Observe that $\Delta(G') < \Delta(G)/2$. We then color the nodes of $G'$ as we did on $G$. It remains to show that the nodes in $G'$ have enough slack. Let's fix a node $v$ in $G'$, then
    \begin{align*}
        \sum_{x \in L_v'} (d_v'(x) + 1) &\geq \sum_{x \in L_v'} (d_v(x) - a_v(x) + 1) + \sum_{x \in L_v \setminus L_v'} (d_v(x) - a_v(x) + 1) \\
        &= \sum_{x \in L_v} (d_v(x) + 1) - \sum_{x \in L_v} (a_v(x)) \\
        &> \deg_G(v) - \widetilde{\deg}_G(v) \\
        &= \deg_{G'}(v)
    \end{align*}
    Whenever the max. degree of such a subgraph $G'$ is a small constant, the problem can be trivially solved in $O(1)$ communication rounds. Since we have seen that the max. degree drops by a factor at least $2$ in each recursion, this has to be repeated $O(\log \Delta)$ times.  
\end{proof}

\end{document}